\documentclass[10pt, a4paper]{article}
\usepackage{comment}
\usepackage{amsmath}
\usepackage{amssymb}
\usepackage{amsthm}
\usepackage{amscd}
\usepackage{subfigure}
\usepackage{multirow}
\usepackage{slashbox}
\usepackage{color}
\usepackage{pst-plot}
\usepackage{graphicx}
\usepackage{indentfirst}
\usepackage[ruled, vlined, linesnumbered]{algorithm2e}
\usepackage{titlesec}
\usepackage[top=25.4mm,  bottom=25.4mm,  left=32.7mm,  right=32.2mm]{geometry}
\usepackage{titlesec}
\usepackage{color}
   \usepackage{amsmath, amssymb}
   \usepackage{latexsym}

\begin{document}

\newtheorem{definition}{\bf Definition}
\newtheorem{notation}{\bf Notation}
\newtheorem{theorem}{\bf Theorem}
\newtheorem{proposition}{\bf Proposition}
\newtheorem{lemma}{\bf Lemma}
\newtheorem{corollary}{\bf Corollary}
\newtheorem{example}{\bf Example}
\newtheorem{remark}{\bf Remark}
\newtheorem{Table}{\bf Table}
\newtheorem{Sentence}{\bf Step}
\newtheorem{Branch}{}
\newtheorem*{problem}{Problem}

\def\T {\ensuremath{\bf{T}}}
\def\N {\ensuremath{\mathbb{N}}}
\def\A {\ensuremath{\bf{A}}}
\def\B {\ensuremath{\bf{B}}}
\def\P {\ensuremath{\bf{P}}}
\def\S {\ensuremath{\mathbb{S}}}
\def\E {\ensuremath{\bf{E}}}
\def\H {\ensuremath{\bf{H}}}
\def\V {\ensuremath{\rm{V}}}
\def\D {\ensuremath{\rm{D}}}
\def\PF {\ensuremath{\bf {PF}}}
\def\TH {\ensuremath{\bf{TH}}}
\def\RS {\ensuremath{\mathbb{T}}}
\def\HCTD {\ensuremath{\tt{HCTD}}}
\def\HPCTD {\ensuremath{\mathrm{HPCTD}}}
\def\CTD {\ensuremath{\mathrm{CTD}}}
\def\PCTD {\ensuremath{\mathrm{PCTD}}}
\def\WUCTD {\ensuremath{\mathrm{WUCTD}}}
\def\RSD {\ensuremath{\mathrm{RSD}}}
\def\FWCTD {\ensuremath{\mathrm{FWCTD}}}
\def\SWCTD {\ensuremath{\mathrm{SWCTD}}}
\def\ARSD {\ensuremath{\tt{RSD}}}
\def\APRSD {\ensuremath{\tt{WRSD}}}
\def\CCTD {\ensuremath{\tt{CTD}}}
\def\ASWCTD {\ensuremath{\tt{SWCTD}}}
\def\ASMPD {\ensuremath{\tt{SMPD}}}
\def\AHPCTD {\ensuremath{\tt{HPCTD}}}
\def\TDU {\ensuremath{\mathrm{TDU}}}
\def\ATDU {\ensuremath{\tt{TDU}}}
\def\RDU {\ensuremath{\mathrm{RDUForZD}}}
\def\ARDU {\ensuremath{\tt{RDUForZD}}}
\newtheorem{Rules}{\bf Rule}

\newcommand{\disc}[1]{\mbox{{\rm disc}$(#1)$}}
\newcommand{\alg}[1]{\mbox{{\rm alg}$(#1)$}}
\newcommand{\SAT}[1]{\mbox{{\rm sat}$(#1)$}}
\newcommand{\SQR}[1]{\mbox{{\rm sqrt}$(#1)$}}
\newcommand{\ideal}[1]{\langle#1\rangle}
\newcommand{\I}[1]{\mbox{{\rm I}$_{#1}$}}
\newcommand{\ldeg}[1]{\mbox{{\rm ldeg}$(#1)$}}
\newcommand{\iter}[1]{\mbox{{\rm iter}$(#1)$}}
\newcommand{\mdeg}[1]{\mbox{{\rm mdeg}$(#1)$}}
\newcommand{\lv}[1]{\mbox{{\rm lv}$_{#1}$}}
\newcommand{\mvar}[1]{\mbox{{\rm mvar}$(#1)$}}
\newcommand{\prem}[1]{\mbox{{\rm prem}$(#1)$}}
\newcommand{\pquo}[1]{\mbox{{\rm pquo}$(#1)$}}
\newcommand{\rank}[1]{\mbox{{\rm rank}$(#1)$}}
\newcommand{\res}[1]{\mbox{{\rm res}$(#1)$}}
\newcommand{\cls}[1]{\mbox{{\rm cls}$_{#1}$}}
\newcommand{\sat}[1]{\mbox{{\rm sat}$(#1)$}}
\newcommand{\sep}[1]{\mbox{{\rm sep}$(#1)$}}
\newcommand{\tail}[1]{\mbox{{\rm tail}$(#1)$}}
\newcommand{\zm}[1]{\mbox{{\rm MZero}$(#1)$}}
\newcommand{\zero}[1]{\mbox{{\rm Zero}$(#1)$}}
\newcommand{\rd}[1]{\mbox{{\rm Red}$(#1)$}}
\newcommand{\map}[1]{\mbox{{\rm map}$(#1)$}}
\newcommand{\op}[1]{\mbox{{\rm op}$(#1)$}}

\title{ {\bf Special Algorithm for Stability Analysis\\ of Multistable Biological
Regulatory Systems}\thanks{This research was partly supported by US National Science Foundation
Grant 1319632, China Scholarship Council, and National Science Foundation
of China Grants 11290141 and 11271034.}}

\author{
Hoon Hong\\\\
{\small \it Department of Mathematics, North Carolina State University, Raleigh
NC 27695, USA}\\\\
Xiaoxian Tang\thanks{Corresponding author}, ~~Bican Xia\\\\
{\small \it LMAM \& School of Mathematical Sciences, Peking University, Beijing
100871, China}\\
}

\date{}
\maketitle
\begin{abstract}
 We consider the problem of  counting (stable) equilibriums of an important
family of algebraic differential equations modeling multistable biological
regulatory systems.   The problem can be solved, in principle,  using real
quantifier elimination
algorithms, in particular real root classification algorithms. However, it
is well known that they can handle only very small cases due to the enormous
computing time requirements. In this paper, we present a special algorithm
which is much more efficient than the general methods.   Its efficiency comes
from the exploitation of certain interesting structures of the family of
differential equations.

~
\end{abstract}
\section{Introduction}\label{sec:1}

Modeling biological networks mathematically as dynamical systems and analyzing
the local and global behaviors of such systems is an important method of
computational biology. The most concerned behaviors of such biological systems
are equilibrium, stability, bifurcations, chaos and so on.

Consider the stability analysis of biological networks modeled by autonomous
systems of differential equations of the form $\dot{{\mathbf x}}={\mathbf
f}\left({\mathbf u},{\mathbf x}\right)$
where ${\mathbf x}=\left(x_1,\ldots,x_n\right)$, $${\mathbf f}\left({\mathbf
u},{\mathbf x}\right)=\left(f_1\left({\mathbf u},x_1,\ldots,x_n\right),\ldots,f_n\left({\mathbf
u},x_1,\ldots,x_n\right)\right)$$ and each $f_k\left({\mathbf u},x_1,\ldots,x_n\right)$
is a rational function in $x_1,\ldots,x_n$ with real coefficients and real
parameter(s) ${\mathbf u}$. We would like to compute a partition of the parametric
space of ${\mathbf u}$ such that, inside every open cell of the partition,
the number of (stable) equilibriums of the system is uniform. Furthermore,
for each open cell, we would like to determine the number of (stable) equilibriums.

Such a problem can be easily formulated as a real quantifier elimination
problem. It is well known that the real quantifier elimination problem can
be carried
out algorithmically.
\cite{tarski1951,collins1975,arnon1988, m1988, mccallum1999,mccallum2001,
Grig88, hong1990a, hong1990b,ch1991,renegar1992a,renegar1992b,
renegar1992c, BPR1996, BPRRoadmap,BPRBook, brown2001a,brown2001b, dss2004,
bm2005,
strzebonski2006,strzebonski2011,brown2012,hs2012,bdemw2013,brown2013}.
 There are several software systems such
as {\tt QEPCAD} \cite{ch1991,hong1992,brown2001a,b}, 
{\tt Redlog} \cite{ds1997}, {\tt Reduce} (in Mathematica) \cite{s2000,s2005}
and {\tt SyNRAC} \cite{ay}. 
Hence, in principle, the stability analysis of  
regulation system
the above system can be carried out automatically using those software systems.
 However, it is also well known that the  complexity \cite{dh1988,BPR1996}
of those algorithms are way beyond current computing capabilities since those
algorithms are for general quantifier elimination problems.  

The stability analysis is a special type of  quantifier elimination problem,
in particular, real root classification. Hence, it would be advisable to
use real root classification algorithms~\cite{yhx,yx}.
In fact, \cite{wx2005issac,wx2005ab}, \cite{nw2008} and \cite{n2012} tackled
the stability analysis problem using {\tt DISCOVERER}~\cite{x}\footnote{{\tt
DISCOVERER} was integrated later in the {\tt RegularChains} package in Maple.
Since then, there are
several improvements on the package from both mathematical and  programming
aspects \cite{cdmmxx}. One can see the command  {\tt RegularChains[ParametricSystemTools][RealRootClassification]}
in any version of Maple that is newer than Maple $13$.}.   They were able
to tackle a specialized simultaneous decision problem ($n=6$ and $c=2$) \cite{cd2002}
in $55,000$ secs \cite{n2012}.  However, the real root classification software
could not go beyond these,  due to enormous computing time/memory requirements.

In this paper, we consider the problem of counting (stable) equilibriums
of an important family of algebraic differential equations modeling multistable
biological regulation systems, called MSRS (see Definition \ref{cd}). In
fact, the family is a straightforward generalization of several interesting
classes of systems in the literature   \cite{cd2002,cd2005,cp2007}.
The family of differential equations has the form
$\dot{{\mathbf x}}={\mathbf f}\left(\sigma, {\mathbf x}\right)$
where ${\mathbf f}$ is a real function determined by
certain real functions $l\left(z\right)$, $g\left(z\right)$, $h\left(z\right)$
and $P\left({\mathbf x}\right)$ and parameterized by a real parameter $\sigma$.
  
We present a special algorithm which is much more efficient than the general
root classification algorithm.
 The efficiency of the special algorithm comes from the exploitation of certain
interesting structures of the differential equation under investigation such
as
\begin{enumerate}
\item[(1)] the eigenvalues of the Jacobian at every equilibrium  are all
 real,
see Theorem \ref{real};

\item[(2)]every  equilibrium of the system is  made up of at most two components,
see
Theorem~\ref{nd};
\item[(3)]the eigenvalues of the Jacobian at every equilibrium have  certain
structures (see Theorems~\ref{dchar} and \ref{ndchar}), aiding the determination
of stability
of an equilibrium (see Corollary~\ref{stable}).
\end{enumerate}
The special algorithm can handle much larger  system than the general root
classification algorithm.
For example, it can handle a specialized simultaneous decision problem ($n=11$
and $c=8$)  in several seconds.

We remark that our work can be viewed as following the numerous efforts in
applying quantifier elimination to tackle problems from various other disciplines
\cite{lazard1988,ls1993,gonzalez1996,dya1997,
jirstrand1997,weispfenning1997,hrsa,hrsb,yxl1999,aw,wx2005issac,
wx2005ab,bnw2006,ggl2006,nw2008,xyz,
sxxz,swak,
n2012,slxzx}.

The paper is organized as follows.
Section \ref{problem} provides a precise statement of the problem.
Section~\ref{review} reviews a general algorithm based on real root classification.
Section \ref{structure} proves several interesting   structures of the problem.
Section \ref{algorithm} gives a special algorithm that exploits the structure
proved  in Section \ref{structure}.
Section \ref{performance} presents
the experimental timings and compares them to those of a general algorithm.
\section{Problem}\label{problem}
In this section, we give a precise and self-contained description  of  the
problem. First we introduce a family of differential equations that we will
be considering.
\begin{definition}[ MultiStable Regulatory System]\label{cd}
A system of ordinary differential equations
\begin{align*}
\frac{dx_1}{dt}&=f_1\left(\sigma,x_1,\ldots,x_n\right) \\
          &\vdots \\
\frac{dx_n}{dt}&= f_n\left(\sigma,x_1,\ldots,x_n\right)
\end{align*}
is called a {\em multistable regulatory system} (MSRS) if $f_k$  has the
following form
\[
f_{k}\left(\sigma,x_1,\ldots,x_n\right)=-l\left(x_k\right)+\sigma \frac{g\left(x_k\right)}{P\left(x_1,\ldots,x_n\right)+h\left(x_k\right)}\]
where
\begin{enumerate}
\item $\sigma$ is a positive parameter;
\item The function $P$ is symmetric, that is, $$P\left(x_{1},\ldots x_{i},\ldots
,x_{j},\ldots ,x_{n}\right)
=P\left( x_{1},\ldots x_{j},\ldots ,x_{i},\ldots ,x_{n}\right)$$ for every
$i,j$;
\item $\forall k\; \forall (x_1,\ldots,x_n) \in {\mathbb R}^n_{>0} \;\;\;P\left(x_1,\ldots,x_n\right)+h\left(x_k\right)>0$;
\item $l\left(z\right)\neq 0$ and for every $\sigma\in {\mathbb R}_{>0}$,
the function $$\sigma\frac{g\left(z\right)}{l\left(z\right)}-h\left(z\right)
$$ has at most
one extreme point on the intended domain of $z$.
\end{enumerate}
\end{definition}

\begin{example}\label{models}
We present several examples of MSRS from cellular differentiation \cite{cd2002,cd2005,cp2007}.
 In fact, the above definition of MSRS is a straightforward generalization
of those differential equations.  \begin{enumerate}
\item{\it Simultaneous decision} \cite{cd2002}.
\[\frac{dx_k}{dt}=-x_{k}+\sigma \frac{1}{1+\Sigma
_{m=1}^{n}x_{m}^{c}-x_{k}^{c}}\]
where the quantities $x_1,...,x_n$ $(\in {\mathbb R}_{>0})$ denote the concentrations
of $n$
proteins,
$c$ ($\in {\mathbb R}_{>0}$) the cooperativity, and $\sigma$ ($\in {\mathbb
R}_{>0}$)  the
strength of unrepressed protein expression, relative to the exponential
decay.
It is easy to verify that it is  a MSRS  with
\begin{align*}
&l\left( z\right)  =z,\;
g\left( z\right)  =1,\;
h\left( z\right)  =-z^{c},\\
&P\left(x_1,\ldots,x_n\right)  =1+\Sigma _{m=1}^{n}x_{m}^{c}.
\end{align*}
The first graph in Figure \ref{fig:n2} shows the graph of $\sigma\frac{g\left(z\right)}{l\left(z\right)}-h\left(z\right)$
for $c=4$ and $\sigma=1$.
\item{\it Mutual inhibition with autocatalysis} \cite{cd2005}.
 \[\frac{dx_k}{dt}=-x_{k}+\alpha +\sigma \frac{x_{k}^{c}}{%
1+\Sigma _{m=1}^{n}x_{m}^{c}}\]
where the quantities $x_1,...,x_n$ $(\in {\mathbb R}_{>0})$ denote the concentrations
of $n$
proteins,
$c$ ($\in {\mathbb R}_{>0}$) the cooperativity,  $\sigma$ ($\in {\mathbb
R}_{>0}$)  the relative speed for transcription/translation, and $\alpha$
($\in {\mathbb R}_{\geq 0}$) the leak expression. It is easy to verify that
it is  a MSRS with\begin{align*}
&l\left( z\right)  =z-\alpha,\;
g\left( z\right)  =z^{c},\;
h\left( z\right)  =0, \\
&P\left(x_1,\ldots,x_n\right) =1+\Sigma _{m=1}^{n}x_{m}^{c}.
\end{align*}
The second graph in Figure \ref{fig:n2} shows the graph of $\sigma\frac{g\left(z\right)}{l\left(z\right)}-h\left(z\right)$
for $\alpha=1$, $c=2$ and $\sigma=1$.
\item{\it bHLH dimerisation} \cite{cd2005,cp2007}.
 \[\frac{dx_k}{dt}=-x_{k}+\sigma \frac{x_{k}^{2}}{\frac{K_2}{a_t^2}
\left(1+\Sigma _{m=1}^{n}x_{m}\right)^{2}+x_{k}^{2}}\]
where the quantities $x_1,...,x_n$ $(\in {\mathbb R}_{>0})$ denote the concentrations
of $n$
proteins, $\sigma$ ($\in {\mathbb
R}_{>0}$)  the relative speed for transcription/translation,  $K_2$ $(\in
{\mathbb R}_{>0})$ the binding constant, and $a_t$ $(\in
{\mathbb R}_{>0})$ the total quantity of activator. It is easy to verify
that it is  a MSRS with\begin{align*}
&l\left( z\right)  =z,\;
g\left( z\right)  =z^2,\;
h\left( z\right)  =z^2,\\
&P\left(x_1,\ldots,x_n\right) =\frac{K_2}{a_t^2} \left(1+\Sigma _{m=1}^{n}x_{m}\right)^{2}.
\end{align*}
The third graph in Figure \ref{fig:n2} shows the graph of $\sigma\frac{g\left(z\right)}{l\left(z\right)}-h\left(z\right)$
for $\sigma=1$.
\end{enumerate}
\begin{figure}[h]
  \centering
 \includegraphics[width=0.36\textwidth]{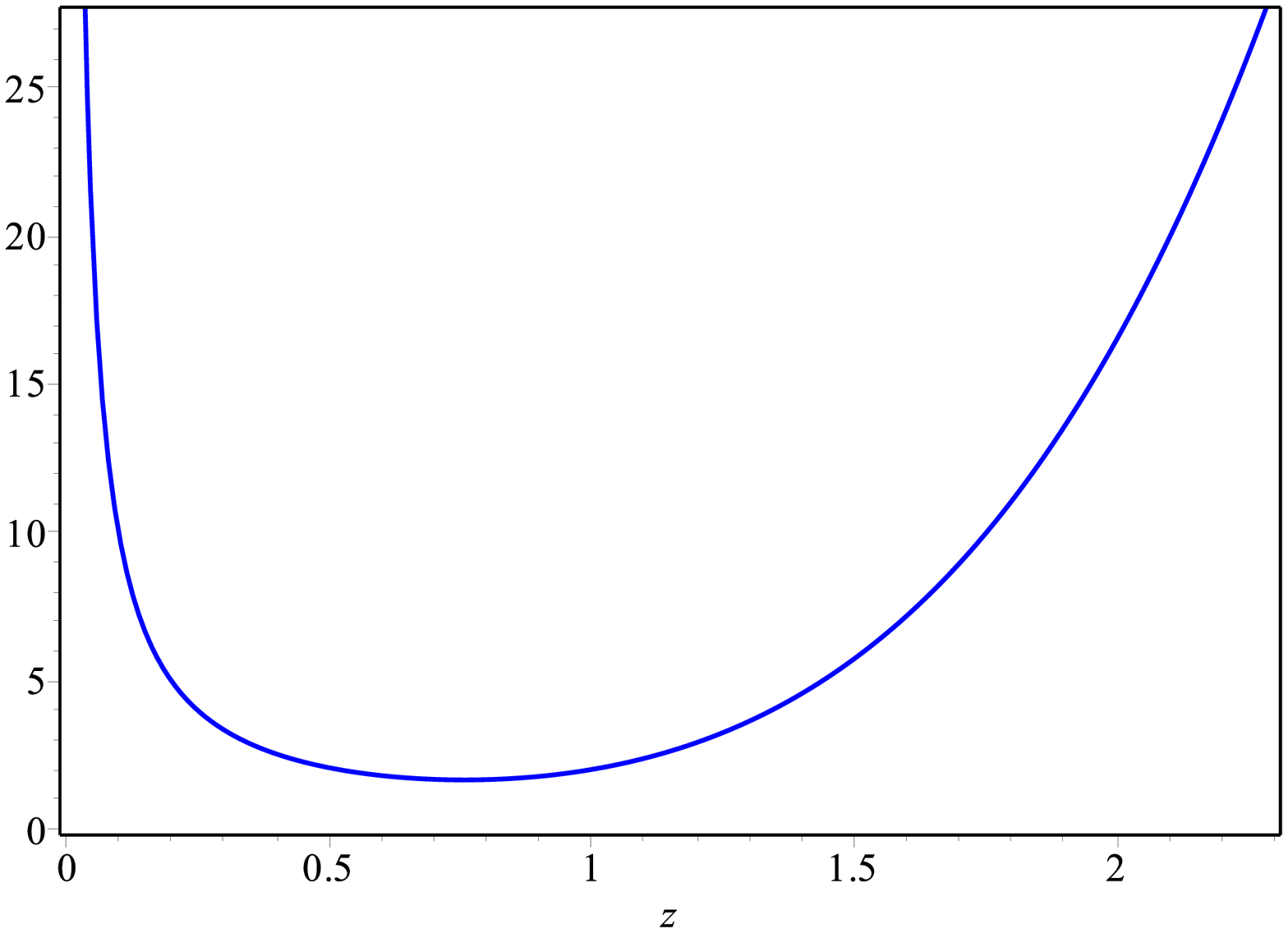}
 \includegraphics[width=0.315\textwidth]{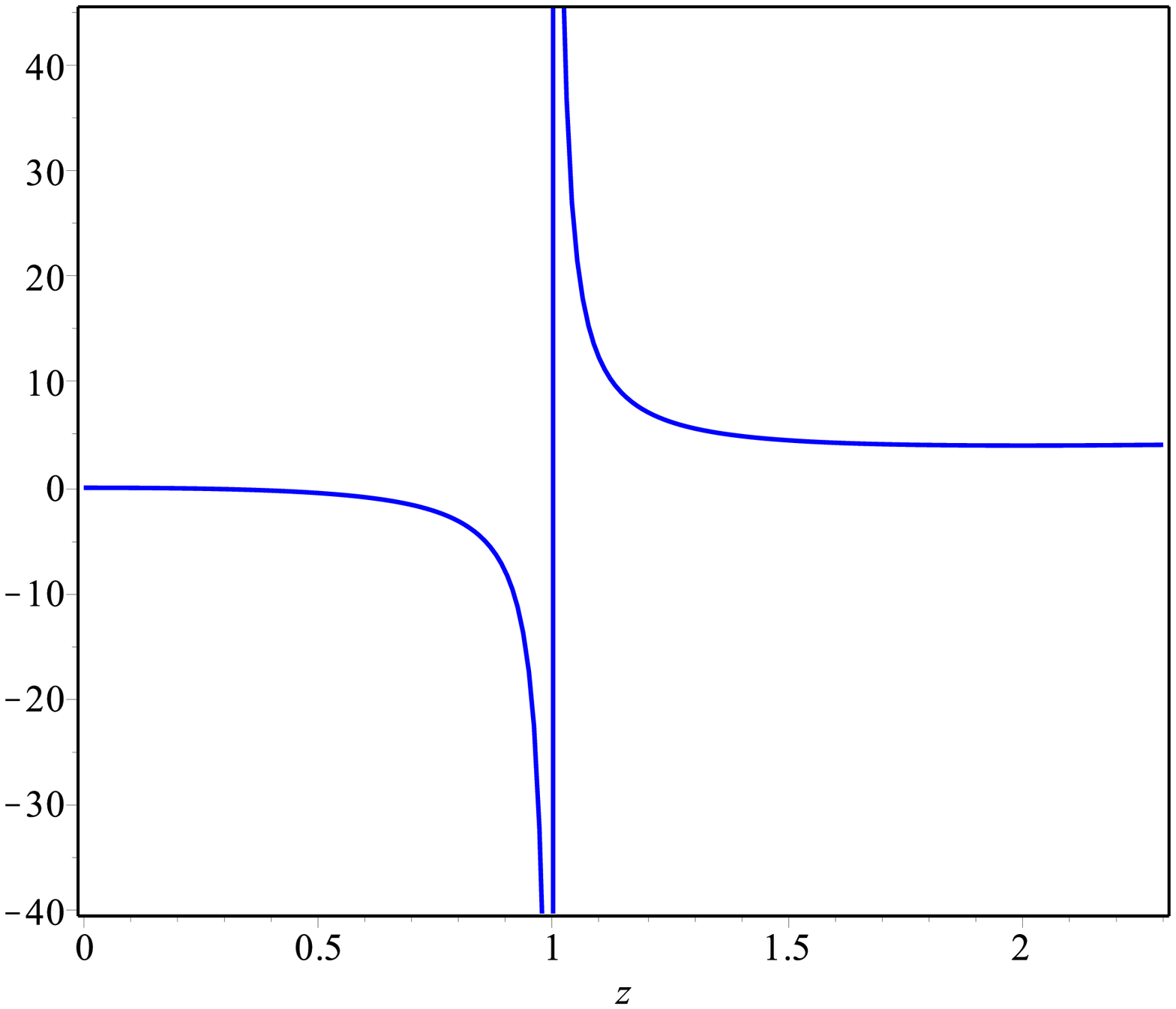}
\includegraphics[width=0.30\textwidth]{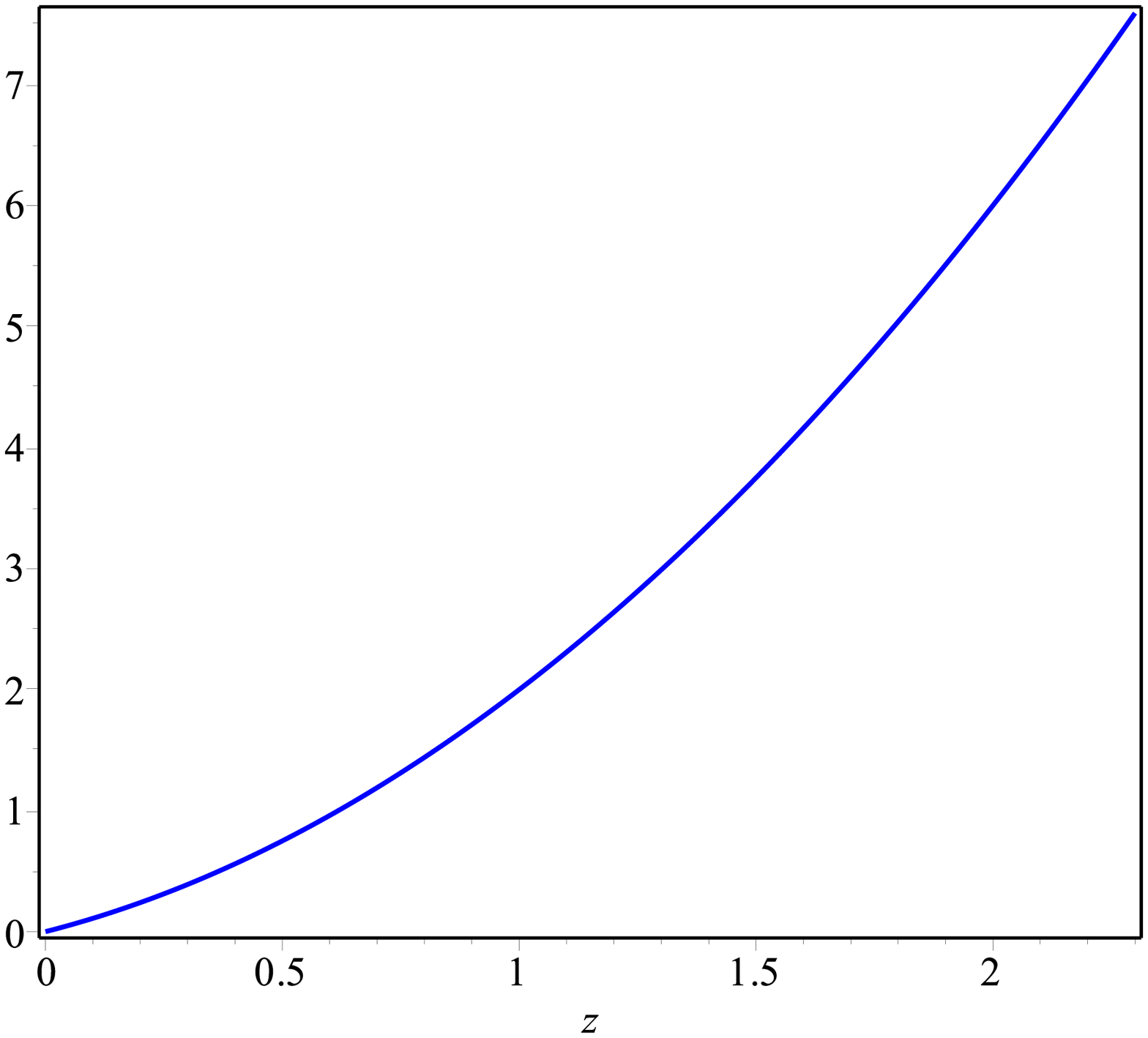}
  \caption{Graphs of $\sigma\frac{g\left(z\right)}{l\left(z\right)}-h\left(z\right)$
for the models in Example \ref{models}}\label{fig:n2}
\end{figure}
\end{example}

\begin{definition}[Equilibrium]
For given $\sigma$, an  ${\mathbf r}\in {\mathbb R}_{>0}^n$ is called an
 {\it equilibrium} if $$f_1\left({\mathbf r}\right)=\cdots=f_n\left({\mathbf
r}\right)=0.$$
\end{definition}

\begin{notation}[Jacobian]
The Jacobian of ${\mathbf f}$ is denoted by
\begin{equation*}
J_{\mathbf f}=\left[
\begin{matrix}
\frac{\partial f_1}{\partial x_1} & \cdots & \frac{\partial
f_1}{\partial
x_n} \\
\vdots & \ddots & \vdots \\
\frac{\partial f_n}{\partial x_1} & \cdots & \frac{\partial
f_n}{\partial
x_n}%
\end{matrix}
\right].
\end{equation*}
\end{notation}
\begin{definition}[Stable]\label{lastable}
An equilibrium ${\mathbf r}$ is called {\it stable} (more precisely, {\em
locally
asymptotically stable}) if all eigenvalues of $J_{\mathbf f}\left({\mathbf
r}\right)$ have strictly
negative real parts. 
\end{definition}

We are ready to state the problem that will be tackled in this paper.
   Informally, the problem is as follows.  For given polynomials $l,g,h$
and $P$, we have a family of MSRS parameterized by $\sigma$. We would like
to
find a partition of $\sigma$ values into several intervals so that  for all~$\sigma$
in each interval the number of (stable) equilibriums is uniform. Furthermore,
for each interval, we would like to determine the number of (stable) equilibriums.
 Now let us state the problem precisely.
\begin{problem} Devise an  algorithm with the following specification.
\medskip
\begin{enumerate}
\item[]{\bf Input:} ${\mathbf f}=\left(f_1,\ldots,f_n\right)\in \left({\mathbb
Q}\left(\sigma,{\mathbf
x}\right)\right)^n$ such
that $\dot{\mathbf x}={\mathbf f}$ is a MSRS
\item[]{\bf Output:}

$B\in \mathbb{Z}[\sigma],$\\
$I_1,\ldots,I_{w-1} \in \mathbb{IQ}_{>0}$\ (that is, closed intervals
with positive rational endpoints) and \\
$\left(e_1,s_1\right), \ldots, \left(e_w,s_w\right)\in  \mathbb{Z}^2_{\geq
0} $ \\
such that
\begin{enumerate}
\item[] $\forall j \in \{1,\ldots,w-1\}, \;B$ has one and only one real root,
say $\sigma_j$, in $I_j$,
\item[] $\sigma_1 < \cdots <\sigma_{w-1}$, and
\item[] $\forall j \in \{1,\ldots,w\}\;\;\forall v\in \left(\sigma_{j-1},
\sigma_j\right)
\;\; E_{v}=e_j\; \wedge S_{v}=s_j$
\end{enumerate}
where \begin{enumerate}
\item[]$\sigma_0=0$, $\sigma_w=\infty$,
\item[]$E_{v}$ ($S_{v}$) denotes the number of (stable) equilibriums
of  $\dot{\mathbf x}={\mathbf f}\left(v,{\mathbf x}\right)$

\end{enumerate}
\end{enumerate}
\end{problem}
\begin{example}\label{result} We illustrate the above input and output specification
by an  example, which is a specific  simultaneous decision model ($n=4$ and
$c=4$) as shown in Example~\ref{models}.

\begin{enumerate}
\item[] {\bf Input:} $f_1,f_2,f_3,f_4$\medskip\\ where
$f_k=-x_k+\frac{\sigma}
{1+x_1^4+x_2^4+x_3^4+x_4^4-x_k^4},\;\;k=1,\ldots,4$

\item[] {\bf Output:}

$B =
\left(42755090541778564453125\sigma^{24}+\cdots-140737488355328\right)\left(\sigma-4\right)^2$,
\medskip\\
$I_1 = [\frac{5}{4}, \frac{21}{16}],\;\; I_2  =[4, 4],$
\medskip\\
$\left(e_1,s_1\right)=\left(1, 1\right),\;\; \left(e_2,s_2\right)=\left(9,
5\right),\;\; \left(e_3,s_3\right)=\left(15, 4\right)$
\end{enumerate}
\medskip
By Definition \ref{cd}, the input system is
\[
\frac{dx_k}{dt}=-x_k+\frac{\sigma}
{1+x_1^4+x_2^4+x_3^4+x_4^4-x_k^4}.\]
The meaning of the  output is as  follows. Let $\sigma_1~ (\approx 1.303331342)$
be the unique positive root of $B\left(\sigma\right)=0$ in $I_1$  and $\sigma_2~(=
4)$ be the
unique positive root of $B\left(\sigma\right)=0$ in $I_2$.
  Then the  system has the following properties:
\begin{enumerate}
\item[(1)]if $0<\sigma<\sigma_1$, then the system has exactly $1$
equilibrium and the equilibrium is stable;
\item[(2)]if $\sigma_1<\sigma<\sigma_{2}$,
then the system has exactly $9$ distinct equilibriums, $5$ of which
are stable;
\item[(3)]if $\sigma_2< \sigma <\infty$, then the system has exactly $15$
distinct
equilibriums, $4$ of which
are stable.
\end{enumerate}
\end{example}

\section{Review of General Algorithm}\label{review}

In this section, we briefly review a general algorithm \cite{wx2005issac,
wx2005ab, nw2008,n2012}
for stability analysis based on real root classification. 
As stated in Section \ref{sec:1}, the general algorithm works for systems
with rational functions and
thus can be applied to solve the Problem posted in last section for MSRS
if all the involved functions,
{\it i.e.}, $l,g,h,P$, are polynomials.

Suppose we are given a system  $\dot{{\mathbf x}}={\mathbf f}\left(\sigma,{\mathbf
x}\right)$ where
$${\mathbf f}\left(\sigma,{\mathbf x}\right)=(f_1(\sigma,x_1,\ldots,x_n),\ldots,f_n(\sigma,x_1,\ldots,x_n))$$
and each
$f_k(\sigma,x_1,\ldots,x_n)$ is a rational function.
A sketch description of the general algorithm may be as follows.
\begin{enumerate}
\item Equate the numerators of all $f_k(\sigma,x_1,\ldots,x_n)$ to $0$, yielding
a system of polynomial equations.
To simplify the notations, we still use $\{f_1=0,\ldots,f_n=0\}$ to denote
the equations.
Note that there may be some constraints on the system. For example, the denominators
of all $f_k$ should be nonzero,
$\sigma$ and some variables should be positive, and so on. Therefore, we
actually obtain a semi-algebraic system. Let us
denote it by ${\mathcal S}$.

\item Compute the Hurwitz determinants $\Delta_1,\ldots,\Delta_n$ of the
Jacobian matrix $J_{\mathbf f}\left(\sigma,{\mathbf x}\right)$.
Let $\det\left(\lambda
I-J_{\mathbf f}\left(\sigma,{\mathbf x}\right)\right)=b_n\lambda^n+b_{n-1}\lambda^{n-1}+\ldots+b_0$
$\left(b_n>0\right)$, then $\Delta_1,\ldots,\Delta_n$ are defined as the
leading principal minors of
\begin{equation*}
\left[
\begin{matrix}
b_{n-1} & b_{n-3} & b_{n-5} & \ldots & b_{n-(2n-1)}\\
b_{n}   & b_{n-2} & b_{n-4} & \ldots & b_{n-(2n-2)}\\
0       & b_{n-1} & b_{n-3} & \ldots & b_{n-(2n-3)}\\
0       & b_{n}   & b_{n-2} & \ldots & b_{n-(2n-4)}\\
0       & 0       & b_{n-1} & \ldots & b_{n-(2n-5)}\\
\vdots  & \vdots  & \vdots  & \vdots & \vdots
\end{matrix}%
\right] _{n\times n}.
\end{equation*}%
By the Routh-Hurwitz Critierion, an equilibrium ${\mathbf r}$ is stable if
and only if
\[\Delta_1({\mathbf r})>0\wedge\cdots\wedge\Delta_n({\mathbf r})>0.\]
Therefore, add the constraints $\Delta_1>0,\ldots,\Delta_n>0$ to ${\mathcal
S}$ and obtain a new system~${\mathcal T}$.

\item Compute the so-called {\em border polynomial} $B(\sigma)$ of the system
${\mathcal T}$. Simply speaking,
$B(\sigma)$ is a polynomial in $\sigma$ satisfying
\[
  \left[\;
     \exists {\mathbf x} 
        \left({\mathbf f}(\sigma,{\mathbf x})=0 \;\wedge\; 
        \det\left(J_{\mathbf f}\left(\sigma,{\mathbf x}\right)\right)
        \cdot
        \prod_{k=1}^n\Delta_k(\sigma,{\mathbf x})=0)\right)
  \;\right]
  \;\;\Longrightarrow\;\; 
  B\left(\sigma\right)=0.
\]
For more details on border polynomials, please refer to \cite{yhx,wx2005issac}.

\item  Because there is only a single parameter $\sigma$, we can take a rational
sample point~$v_j$ in the open interval 
$(\sigma_j,\sigma_{j+1})$ for all $j\; \left(0\le j\le w-1\right)$ by isolating
the distinct positive roots $\sigma_1,\ldots,\sigma_{w-1}$ of 
$B(\sigma)=0,$ where $\sigma_0=0$ and $\sigma_{w}=+\infty$.

\item For each sample point $v_j$, substitute~$v_j$ for~$\sigma$ in~${\mathcal
S}$ and ${\mathcal T}$, respectively, yielding two new constant systems~${\mathcal
S}(v_j)$ and ${\mathcal T}(v_j)$. By real solution counting (or isolating)
of~${\mathcal S}(v_j)$ and ${\mathcal T}(v_j)$, respectively, we obtain the
number of equilibriums and the number of stable equilibriums of the original
system at $v_j$, respectively. By the property of $B(\sigma)$, the number
of (stable) equilibriums of the original system at~$v_j$ equals  the number
of (stable) equilibriums of the original system at any $\sigma\in (\sigma_j,\sigma_{j+1})$.
\end{enumerate}

In general, the Hurwitz determinants may be huge and thus computing them
is very time-consuming. Furthermore, huge Hurwitz determinants may cause
it infeasible in practice to compute the border polynomial of system ${\mathcal
T}$.

\section{Structure}\label{structure}
In this section, we describe certain special structures of the multi-stable
regulatory system that we will exploit in order to develop an efficient special
algorithm.
Before we plunge into the details, we  first provide an overview of the 
 special structures:
\begin{enumerate}
\item[(1)] the eigenvalues of the Jacobian at every equilibrium  are all
 real,
see Theorem~\ref{real};

\item[(2)]every  equilibrium of the system is  made up of at most two components,
see
Theorem~\ref{nd};
\item[(3)]the eigenvalues of the Jacobian at every equilibrium have  certain
nice structures, simplifying the stability analysis,  see Theorems \ref{dchar}
and \ref{ndchar} and Corollary~\ref{stable}.
\end{enumerate}

\noindent Now, we plunge into the technical details. In the discussion below,
when we say ``(stable) equilibrium", we mean (stable) equilibrium of a MSRS
 $\dot{{\mathbf x}}={\mathbf f}\left(\sigma,{\mathbf
x}\right)$.
We will use the following notations throughout this section:
\begin{align*}
a\left(\sigma,z\right) &=\sigma\frac{ g\left(z\right)}{l\left(z\right)}-h\left(z\right)
,\\
D_k\left({\mathbf x}\right) &=-\frac{P\left({\mathbf x}\right)+h\left(x_k\right)}{l\left(x_k\right)}.
\end{align*}
It is easy to see that
\[f_k\left({\mathbf x}\right)=\frac{P\left({\mathbf x}\right)-a\left(\sigma,x_k\right)}{D_k\left({\mathbf
x}\right)}.\]

\begin{theorem}[Real eigenvalues]
\label{real}
If $\mathbf{r}$ is an equilibrium, then
every eigenvalue of $J_{\mathbf f}\left(\mathbf{r}\right)$ is real.
\end{theorem}

\begin{proof}
Let $\mathbf{r}$ be an equilibrium and $A=J_{\mathbf f}\left(\mathbf{r}\right)$.
For every $k$, let
\[ N_{k}({\mathbf {x}}) =P\left({\mathbf x}\right)-a\left(\sigma, x_{k}\right).\]
Then for any $i,j$,
\begin{equation*}
A_{i,j}=%
\begin{cases}
\frac{\partial f_{i}}{\partial x_{i}}\left(\mathbf{r}\right) & i=j \\ \\
\frac{\frac{\partial N_{i}}{\partial x_{j}}\left(\mathbf{r}\right)D_{i}\left(
\mathbf{r}%
\right) -N_{i}\left(\mathbf{r}\right)\frac{\partial D_{i}}{\partial x_{j}}\left(\mathbf{r}\right)}{%
D_{i}\left( \mathbf{r}\right) ^{2}} & i\neq j%
\end{cases}%
.
\end{equation*}%
Since $\mathbf{r}$ is an equilibrium, we have $N_{i}\left(\mathbf{r}\right)=0$
for any
$i$. Hence,
\begin{equation*}
A_{i,j}=%
\begin{cases}
\frac{\partial f_{i}}{\partial x_{i}}\left(\mathbf{r}\right) & i=j \\ \\
\frac{\frac{\partial N_{i}}{\partial x_{j}}\left(\mathbf{r}\right)}{D_{i}\left(
\mathbf{%
r}\right) } & i\neq j%
\end{cases}%
=%
\begin{cases}
\frac{\partial f_{i}}{\partial x_{i}}\left(\mathbf{r}\right) & i=j \\ \\
\frac{\frac{\partial P}{\partial x_{j}}\left(\mathbf{r}\right)}{D_{i}\left(
\mathbf{r}%
\right) } & i\neq j%
\end{cases}.%
\end{equation*}%
Let $E$ be the $n\times n$ diagonal matrix such that
\begin{equation*}
E_{i,i}=\frac{\frac{\partial P}{\partial x_{i}}\left(\mathbf{r}\right)}{\Pi
_{k\neq
i}D_{k}\left( \mathbf{r}\right) }.
\end{equation*}
Let $C=EA$. Then for any $i,j$ such that $i\neq j$, we have
\begin{align*}
C_{i,j} =&E_{i,i}A_{i,j}=\frac{\frac{\partial P}{\partial x_{i}}(
\mathbf{r})}{\Pi _{k\neq i}D_{k}\left( \mathbf{r}\right)}\cdot 
\frac{\frac{\partial P}{\partial x_{j}}(
\mathbf{r})}{D_{i}\left( \mathbf{r}\right) } =\frac{\frac{\partial P}{\partial
x_{i}}\left(\mathbf{r}\right)
\frac{\partial P}{\partial x_{j}}\left(\mathbf{r}\right)}{\Pi _{k=1}^nD_{k}\left(
\mathbf{r}
\right) } ,\\
C_{j,i} =&E_{j,j}A_{j,i}=\frac{\frac{\partial P}{\partial x_{j}}(
\mathbf{r})}{\Pi _{k\neq j}D_{k}\left( \mathbf{r}\right) }\cdot
\frac{\frac{\partial P}{\partial x_{i}}(
\mathbf{r})}{D_{j}\left( \mathbf{r}\right) } =\frac{\frac{\partial P}{\partial
x_{j}}\left(\mathbf{r}\right)\frac{
\partial P}{\partial x_{i}}\left(\mathbf{r}\right)}{\Pi _{k=1}^nD_{k}\left(
\mathbf{r}
\right)}.
\end{align*}
Thus $C_{i,j}=C_{j,i}$. Hence $C$ is a real symmetric matrix.

Let $\lambda $ be an eigenvalue of $A$ and $\alpha $ a corresponding
eigenvector, namely $A\alpha =\lambda \alpha $. Then
$C\alpha =EA\alpha =\lambda E\alpha.$ By taking conjugate transpose, we have
\begin{equation*}
\alpha ^{\ast }C^{\ast }=\lambda ^{\ast }\alpha ^{\ast }E^{\ast }.
\end{equation*}%
Since both $E$ and $C$ are real symmetric, we have
$\alpha ^{\ast }C=\lambda ^{\ast }\alpha ^{\ast }E.$
Therefore,
$\alpha ^{\ast }C\alpha =\lambda ^{\ast }\alpha ^{\ast }E\alpha$
and hence
\begin{equation*}
\lambda \alpha ^{\ast }E\alpha =\lambda ^{\ast }\alpha ^{\ast }E\alpha .
\end{equation*}%
Since $\alpha ^{\ast }E\alpha $ is non-zero, we have $\lambda =\lambda
^{\ast }$. In other words, $\lambda $ is real.
\end{proof}

\begin{theorem}[Structure of equilibrium]
\label{nd} Let ${\mathbf r}=\left(r_1,\ldots,r_n\right)$ be an equilibrium.
The components of  ${\mathbf r}$ consist of at most two different numbers.
\end{theorem}

\begin{proof}

For every  $k$, we have%
\begin{align*}
f_k\left({\mathbf r}\right)=\frac{P\left({\mathbf r}\right)-a\left(\sigma,r_k\right)}{D_k\left({\mathbf
r}\right)} & =0.
\end{align*}
Thus%
\[
a\left(\sigma,  r_{1}\right)  =\cdots=a\left(\sigma,  r_{n}\right)  =P\left({\mathbf
r}\right).
\]
  Note that, for every $\sigma$,  the function $a\left(\sigma, z\right)$
has at most
one extreme point  for $z$ over ${\mathbb R}_{>0}$ by Definition \ref{cd}.
Thus for every real number $\varrho$,  the equation
$a\left(\sigma,  z\right)  =\varrho$ has at most two different positive solutions
in $z$.
Hence $r_1,\ldots,r_n$ consist of at most two different positive numbers.
\end{proof}

\noindent
From now on, we will say that an equilibrium ${\mathbf r}$ is \emph{diagonal}
if $r_1=\cdots=r_n$.

\begin{theorem}[Characteristic polynomial for diagonal equilibrium ]
\label{dchar} Let ${\mathbf r}$ be a diagonal equilibrium $\left(q,\ldots,q\right)$.
Then
\begin{equation*}
\det\left(\lambda I-J_{\mathbf f}\left({\mathbf r}\right)\right)=\left(\lambda
-G_1\right)^{n-1}\left(\lambda -G_2\right).
\end{equation*}%
where
\[G_1 = \tau-\xi,\]
\[G_2 = \tau+(n-1)\xi.\]
where again
\[\tau   = \frac{\partial f_n}{\partial x_n}\left({\mathbf r}\right),\;\;\;\;
\xi    = \frac{\frac{\partial P}{\partial x_{n-1}}}{D_n}\left({\mathbf r}\right).\]
\end{theorem}
\begin{proof}
Note for any $i,j$,
\begin{align*}
 f_i(x_1,\ldots,x_i,\ldots,x_j,\ldots,x_n)=f_j(x_1,\ldots,x_j,\ldots,x_i,\ldots,x_n),\\
 P(x_1,\ldots,x_i,\ldots,x_j,\ldots,x_n)=P(x_1,\ldots,x_j,\ldots,x_i,\ldots,x_n).
\end{align*}
Thus,
\begin{align*}
\frac{\partial f_i}{\partial x_i}(x_1,\ldots,x_i,\ldots,x_j,\ldots,x_n)=\frac{\partial
f_j}{\partial
x_j}(x_1,\ldots,x_j,\ldots,x_i,\ldots,x_n),\\
\frac{\partial P}{\partial x_i}(x_1,\ldots,x_i,\ldots,x_j,\ldots,x_n)=\frac{\partial
P}{\partial
x_j}(x_1,\ldots,x_j,\ldots,x_i,\ldots,x_n).
\end{align*}
Hence,
\begin{align*}
\frac{\partial f_i}{\partial x_i}\left({\mathbf r}\right)=\frac{\partial
f_i}{\partial x_i}(q,\ldots,q)=\frac{\partial f_j}{\partial
x_j}(q,\ldots,q)=\frac{\partial f_j}{\partial
x_j}\left({\mathbf r}\right),\\
\frac{\partial P}{\partial x_i}\left({\mathbf r}\right)=\frac{\partial P}{\partial
x_i}(q,\ldots,q)
=\frac{\partial P}{\partial x_j}(q,\ldots,q)=\frac{\partial P}{\partial x_j}\left({\mathbf
r}\right).
\end{align*}
Note also for any $i,j$,
\[D_i\left({\mathbf r}\right)=D_i(q,\ldots,q)=D_j(q,\ldots,q)=D_j\left({\mathbf
r}\right).\]
Therefore
\begin{equation*}
\begin{array}{ccc}
J_{\mathbf f}\left({\mathbf r}\right)=\left[
\begin{matrix}
\tau & \xi & \ldots & \xi \\
\xi & \tau & \ldots & \xi \\
\vdots & \vdots & \ddots & \vdots \\
\xi &\xi & \ldots & \tau%
\end{matrix}%
\right] _{_{n\times n}}
\end{array}.
\end{equation*}
Note
\begin{equation*}
J_{\mathbf f}\left({\mathbf r}\right)
=\left(\tau-\xi\right) I+\xi u^{T}u.
\end{equation*}
where $u=\left[
\begin{array}{ccc}
1 & \cdots & 1%
\end{array}%
\right] .$ Hence,
\begin{align*}
&\det\left(\lambda I-J_{\mathbf f}\left({\mathbf r}\right)\right)\\
&=\det \left( \lambda I-\left( \tau-\xi\right) I-\xi u^{T}u\right) \\
& =\det \left( \left( \lambda -\left(\tau-\xi\right)\right) I-\xi u^{T}u\right)
\\
& =\left(
\lambda -\left(\tau-\xi\right)\right) ^{n}\det \left( I-\frac{\xi}{\lambda
-\left(\tau-\xi\right)}u^{T}u\right)
 \\
 &=\left(
\lambda -\left(\tau-\xi\right)\right) ^{n}\left( 1-\frac{\xi}{\lambda -\left(\tau-\xi\right)}uu^{T}\right)\;
\text{(Sylvester's determinant theorem)}
\\
 &=\left( \lambda -\left(\tau-\xi\right)\right) ^{n}\left( 1-\frac{\xi}{\lambda
-\left(\tau-\xi\right)}n\right)  \\
& =\left(\lambda -\left(\tau-\xi\right)\right)^{n-1}\left(\lambda -\left(\tau+(n-1)\xi\right)
\right) \\
 &= \left(\lambda -G_1\right)^{n-1}\left(\lambda -G_2\right).
\end{align*}
\end{proof}

\begin{theorem}[Characteristic polynomial for non-diagonal equilibrium]
\label{ndchar} Let  ${\mathbf r}$ be a non-diagonal equilibrium. Let $p$
and $q$ appear in ${\mathbf r}$ respectively $i$
times and $n-i$ times, where $1\leq i\leq \lfloor\frac{n}{2}\rfloor$.  Then
\begin{equation*}
\det\left(\lambda I-J_{\mathbf f}\left({\mathbf r}\right)\right)=\left(\lambda-G_1\right)^{n-i-1}\left(\lambda-G_2\right)^{i-1}\left(\lambda^2-G_3\lambda+G_4\right),
\end{equation*}
where
\[G_1 =\tau-\xi,\]
\[G_2=\beta-\gamma,\]
\[G_3=\beta+\tau+\left(i-1\right)\gamma+\left(n-i-1\right)\xi,\]
\[G_4 =\left(\beta+\left(i-1\right)\gamma\right)\left(\tau+\left(n-i-1\right)\xi\right)-i\left(n-i\right)\mu\nu,\]
where again
\[
\beta  = \frac{\partial f_1}{\partial x_1}\left({\mathbf r}\right),\;\;
\tau   = \frac{\partial f_n}{\partial x_n}\left({\mathbf r}\right),\;\;
\gamma = \frac{\frac{\partial P}{\partial x_2}}{D_1}\left({\mathbf r}\right),\;\;
\xi    = \frac{\frac{\partial P}{\partial x_{n-1}}}{D_n}\left({\mathbf r}\right),\;\;
\mu    = \frac{\frac{\partial P}{\partial x_n}}{D_1}\left({\mathbf r}\right),\;\;
\nu    = \frac{\frac{\partial P}{\partial x_1}}{D_n}\left({\mathbf r}\right).
\]
\end{theorem}
\begin{proof}
Without loss of generality,  suppose that $r_{1}=\cdots =r_{i}=p$ and $r_{i+1}=\cdots
=r_{n}=q$. By symmetry, we have
\begin{equation*}
J_{\mathbf f}\left({\mathbf r}\right)=\left[
\begin{matrix}
E & S \\
T & F%
\end{matrix}%
\right] _{n\times n},
\end{equation*}%
where
\begin{equation*}
\begin{array}{ccc}
E=\left[
\begin{matrix}
\beta & \gamma  & \ldots & \gamma \\
\gamma & \beta & \ldots & \gamma \\
\vdots & \vdots & \ddots & \vdots \\
\gamma & \gamma & \ldots & \beta%
\end{matrix}%
\right] _{_{i\times i}} &  & F=\left[
\begin{matrix}
\tau & \xi & \ldots & \xi \\
\xi & \tau & \ldots & \xi \\
\vdots & \vdots & \ddots & \vdots \\
\xi & \xi & \ldots & \tau%
\end{matrix}%
\right] _{(n-i)\times (n-i)} \\
&  &  \\
S=\mu\left[
\begin{matrix}
1 & 1 & \ldots & 1 \\
1 & 1 & \ldots & 1 \\
\vdots & \vdots & \ddots & \vdots \\
1 & 1 & \ldots & 1%
\end{matrix}%
\right] _{i\times (n-i)} &  & T=\nu\left[
\begin{matrix}
1 & 1 & \ldots & 1 \\
1 & 1 & \ldots & 1 \\
\vdots & \vdots & \ddots & \vdots \\
1 & 1 & \ldots & 1%
\end{matrix}%
\right] _{(n-i)\times i}%
\end{array}.
\end{equation*}%
From Laplace's Theorem, we have%
\begin{equation*}
\det\left(\lambda I-J_{\mathbf f}\left({\mathbf r}\right)\right)=\left(-1\right)^{2\left(1+2+\ldots
+i\right)}\det\left(\lambda I-E\right)\det
\left(\lambda I-F\right)+\Sigma _{k=1}^{i}\Sigma _{\omega=1}^{n-i}M_{k,\omega}A_{k,\omega},
\end{equation*}%
where $M_{k,\omega}$ is the minor of $\lambda I-J_{\mathbf f}\left({\mathbf
r}\right)$ consisting of the first
$i$
rows and the columns indexed by
\begin{equation*}
1,2,\ldots ,k-1,k+1,\ldots ,i,i+\omega
\end{equation*}%
and $A_{k,\omega}$ is the cofactor of $M_{k,\omega}$. By the same reasoning
as that
in
the proof of Theorem~\ref{dchar}, we have
\begin{equation*}
\det\left(\lambda I-E\right)=\left( \lambda -\left(\beta+(i-1)\gamma\right)\right)\left(\lambda-G_2\right)^{i-1}
\end{equation*}%
and
\begin{equation*}
\det\left(\lambda I-F\right)=
\left( \lambda -\left(\tau+(n-i-1)\xi\right)\right)\left(\lambda-G_1\right)^{n-i-1}.
\end{equation*}%
It is not difficult to check that
\begin{align*}
M_{k,\omega} &=\left(-1\right)^{i-k+1}\mu \left(\lambda-G_2\right)^{i-1}
,\\
A_{k,\omega} &=\left(-1\right)^{2(1+2+\cdots +i)-k+2\omega+i}\nu \left(\lambda-G_1\right)^{n-i-1}.
\end{align*}%
Hence
\begin{equation*}
\det\left(\lambda I-J_{\mathbf f}\left({\mathbf r}\right)\right)=\left(\lambda-G_1\right)^{n-i-1}\left(\lambda-G_2\right)^{i-1}\left(\lambda^2-G_3\lambda+G_4\right).
\end{equation*}
\end{proof}

\begin{corollary}[Stability of equilibrium]
\label{stable} Let ${\mathbf r}$ be an equilibrium.  Then

\begin{enumerate}
\item[(1)]Case:  ${\mathbf r}$ is diagonal $\left(q,\ldots,q\right)$. Then
${\mathbf r}$   is stable if and only if
\[G_1 <0 \land G_2 <0,\]
where
$G_1$ and $G_2$ are defined as in Theorem~\ref{dchar}.
\item[(2)] Case: ${\mathbf r}$ is non-diagonal such that $p$ appears once
and $q$ appears $n-1$ times. Then 
\begin{enumerate}
\item[(2a)]if $n=2$, then 
${\mathbf r}$   is stable if and only if
\[G_3 <0 \land G_4 >0;\]
\item[(2b)]if $n>2$, then
${\mathbf r}$   is stable if and only if
\[G_1 <0 \land G_3 <0 \land G_4 >0,\]
\end{enumerate}
where $G_1,G_3,G_4$ are defined as in Theorem \ref{ndchar}.
  \item[(3)] Case: ${\mathbf r}$ is non-diagonal such that $p$ appears $i$
times and $q$ appears $n-i$ times where $2\leq i\leq \lfloor\frac{n}{2}\rfloor$.
 Then ${\mathbf r}$   is stable if and only if
\[G_1 <0 \land G_2 <0 \land G_3 <0 \land G_4 >0,\]
where $G_1,G_2,G_3,G_4$ are defined as in Theorem \ref{ndchar}.
\end{enumerate}
\end{corollary}

\begin{proof}
\mbox{}

\begin{enumerate}
\item[(1)]Case:  ${\mathbf r}$ is diagonal $\left(q,\ldots,q\right)$.  From
Theorem~\ref{dchar}, the eigenvalues  of $J_{\mathbf f}\left({\mathbf r}\right)$
are%
\begin{align*}
& \lambda _{1}=\cdots =\lambda _{n-1}=G_1, \\
& \lambda _{n}=G_2.
\end{align*}%
 From Definition \ref{lastable}, the conclusion follows immediately.
\item[(2)] Case: ${\mathbf r}$ is non-diagonal such that $p$ appears once
and $q$ appears $n-1$ times.  
\begin{enumerate}
\item[(2a)]If $n=2$, from Theorem \ref{ndchar}, $\lambda_1$ and $\lambda_2$,
 the eigenvalues of $J_{{\mathbf
f}}\left({\mathbf r}\right)$, 
are the two solutions of $\lambda^2-G_3\lambda+G_4=0$.
Note
\[\lambda_1+\lambda_2=G_3,\]
\[\lambda_1\lambda_2=
G_4.\]
By Theorem
\ref{real},  both $\lambda_{1}$ and $\lambda_2$ are real. Hence,
$\lambda_{1}<0$ and $\lambda_{2}<0$ if and only if $\lambda_{1}+\lambda_{2}<0$
and $%
\lambda_{1}\lambda_{2}>0$.  From Definition \ref{lastable}, the conclusion
follows immediately.
\item[(2b)]If $n>2$, from Theorem \ref{ndchar}, the eigenvalues of $J_{{\mathbf
f}}\left({\mathbf r}\right)$ are \[
\lambda_1=\cdots=\lambda_{n-2}=G_{1}\]
and
\begin{center}
$\lambda_{n-1}$ and $\lambda_n$
are the two solutions of $\lambda^2-G_3\lambda+G_4=0$.
\end{center} Note
\[\lambda_{n-1}+\lambda_n=G_3,\]
\[\lambda_{n-1}\lambda_{n}=
G_4.\]
By Theorem
\ref{real},  both $\lambda_{n-1}$ and $\lambda_n$ are real. Hence,
$\lambda_{n-1}<0$ and $\lambda_{n}<0$ if and only if $\lambda_{n}+\lambda_{n-1}<0$
and $%
\lambda_{n-1}\lambda_{n}>0$.  From Definition \ref{lastable}, the conclusion
follows immediately.
\end{enumerate}
\item[(3)] Case: ${\mathbf r}$ is non-diagonal such that $p$ appears $i$
times and $q$ appears $n-i$ times where $2\leq i\leq \lfloor\frac{n}{2}\rfloor$.
From Theorem \ref{ndchar}, the eigenvalues of $J_{f}\left({\mathbf r}\right)$
are
\[
\lambda_{1}=\cdots=\lambda_{n-i-1}=G_{1},\]
\[
\lambda_{n-i}=\cdots=\lambda_{n-2}=G_{2}\]
and
\begin{center}
$\lambda_{n-1}$ and $\lambda_n$
are the two solutions of $\lambda^2-G_3\lambda+G_4=0$.
\end{center} Note
\[\lambda_{n-1}+\lambda_n=G_3,\]
\[\lambda_{n-1}\lambda_{n}=G_4.\]
By Theorem
\ref{real},  both $\lambda_{n-1}$ and $\lambda_n$ are real. Hence, $\lambda_n<0$
and $\lambda_{n-1}<0$ if and only if $\lambda_{n}+\lambda_{n-1}<0$ and $\lambda_{n-1}\lambda_{n}>0$.
From Definition \ref{lastable}, the conclusion follows immediately.
\end{enumerate}
\end{proof}

\section{Special Algorithm}\label{algorithm}

In this section, we present  algorithms for the problem posed in Section~\ref{problem},
that exploits  several special  structures   proved in Section~\ref{structure}.
 The description of the main algorithm is given in Algorithm~\ref{sa}.  
It is high-level in that it does not specify implemental details. Below we
will explain the main ideas underlying  the sub-algorithms and the main algorithm.

\begin{itemize}
\item {\bf Algorithm \ref{NDE} ({\tt NonDiagonalEquilibrium})}: The correctness
of the algorithm follows from the symmetry of $\dot{\mathbf x}={\mathbf f}$
and Theorem \ref{ndchar}.
\item {\bf Algorithm \ref{DE} ({\tt DiagonalEquilibrium})}: The correctness
of the algorithm follows from the symmetry of $\dot{\mathbf x}={\mathbf f}$
and Theorem \ref{dchar}.
\item {\bf Algorithm \ref{EC} ({\tt EquilibriumCounting})}: Given ${\mathbf
f}$ satisfying the conditions in Definition~\ref{cd},
and a real number $v$, we compute $E_v\;(S_v)$, the number of (stable) equilibriums
of $\dot{\mathbf x}={\mathbf f}\left(v,{
\mathbf x}\right)$.
To this purpose,  we transform the  $n$--dimensional system $\dot{\mathbf
x}={\mathbf f}$ into several $2$--dimensional
systems by Algorithms \ref{DE} and \ref{NDE}, determine the stability easily
by Corollary
\ref{stable} and count the number of (stable) equilibriums by symmetry. 
See more details below.
\begin{itemize}
\item{\bf Lines \ref{ECDE}--\ref{nsde}}: We count the number of diagonal
equilibriums and determine the stability of the diagonal equilibriums by
Corollary \ref{stable}-(1).
\item{\bf Lines \ref{ECNDE}--\ref{nscondition2}}: We are preparing to  count
the number of    non-diagonal equilibriums.
 If $i=1$ and $n=2$,  we determine the stability of a non-diagonal equilibrium
by Corollary
\ref{stable}-(2a). If $i=1$ and $n>2$,  we determine the stability of a non-diagonal
equilibrium
by Corollary
\ref{stable}-(2b). 
 If $i\neq1$,  we determine the stability
of a non-diagonal equilibrium by Corollary
\ref{stable}-(3).

\item{\bf Lines \ref{checksymmetry}--\ref{compute2}}: We  compute the number
of (stable) equilibriums by combining the results computed by {\bf Lines
\ref{ECNDE}}--{\bf \ref{nscondition2}} together. In fact, by the symmetry
of  $\dot{\mathbf x}={\mathbf f}$, for every $i$ $(i=1,\ldots,\lfloor \frac{n}{2}\rfloor)$,
 if the system $\sigma=v\wedge F_1=0\wedge F_2=0\wedge
p\neq q$ has $\tilde{e}_i$  positive solutions, then
\begin{enumerate}
\item[(a)]if $i=\frac{n}{2}$,
the system $\sigma=v\wedge F_1=0\wedge F_2=0\wedge p\neq q$ is symmetric
and thus $\tilde{e}_i$
is even and the system $\dot{\mathbf x}={\mathbf f}$ has $\frac{\tilde{e}_i}{2}\cdot
\binom{n}{i}$ non-diagonal equilibriums.
 \item[(b)]if  $i\neq \frac{n}{2}$, the system $\dot{\mathbf x}={\mathbf
f}$ has $\tilde{e}\cdot \binom{n}{i}$ non-diagonal equilibriums.
\end{enumerate}
Similarly, we count the number of stable equilibriums.
\end{itemize}
\item {\bf Algorithm \ref{CP} ({\tt CriticalPolynomial})}: Given ${\mathbf
f}$ satisfying the conditions in Definition~\ref{cd}, we compute a polynomial
$B\left(\sigma\right)$ such that every ``critical" $\sigma$ value of the
system $\dot{\mathbf x}={\mathbf f}$ is a root of  $B\left(\sigma\right)=0$.
By  the ``critical" values, we mean that the number of the (stable) equilibriums
 of the system changes only when $\sigma$ passes through those values. Note
that the number of the (stable) equilibriums changes only when an eigenvalue
of the Jacobian  vanishes. In diagonal case,  by Algorithm \ref{DE}, an eigenvalue
vanishes if and only if $G_1G_2=0$, see {\bf Lines
\ref{DECP1}--\ref{DECP2}}.  In non-diagonal case, by Algorithm~\ref{NDE},
if an eigenvalue vanishes then $G_1G_2G_3G_4=0$, see {\bf Lines
\ref{NDECP1}--\ref{NDECP2}}.
\item {\bf Algorithm \ref{sa} ({\tt EquilibriumClassification (Special algorithm
for MSRS)})}:

\begin{itemize}
\item {\bf Lines \ref{cp}--\ref{sample}}: By Algorithm \ref{CP}, we compute
$B\left(\sigma\right)$ and isolate the real roots of
$B\left(\sigma\right)=0$.   Note that for all
$\sigma$ in each open interval determined by $B\left(\sigma\right)\neq
0$, the number of (stable) equilibriums is uniform. Thus  we sample  one
rational number $v_i$ from each open interval.
\item {\bf Lines \ref{loop2}--\ref{endloop}}: In this loop,  we compute $e_j\
(s_j)$, the number of (stable) equilibriums
for $\sigma=v_j$ by Algorithm \ref{EC}. We also collect all root isolation
intervals  containing
the ``critical" $\sigma$ values.
Recall that  a  root of $B$ may not be critical, although $B$ vanishes at
every critical $\sigma$ value.
So we  check whether a root
of $B\left(\sigma\right)=0$ is critical or not by {\bf Lines \ref{begincheck}--\ref{endcheck}}.
\end{itemize}
\end{itemize}

\begin{algorithm}\label{sa}
\caption{{\tt EquilibriumClassification (Special algorithm for MSRS)}}
\KwIn{
\begin{enumerate}
\item[]${\mathbf f}=\left(f_1,\ldots,f_n\right)\in \left({\mathbb Q}\left(\sigma,{\mathbf
x}\right)\right)^n$ such
that $\dot{\mathbf x}={\mathbf f}$ is a MSRS
\end{enumerate}
}
\KwOut{
\begin{enumerate}
\item[] $B\in \mathbb{Z}[\sigma]$,\\
$I_1,\ldots,I_{w-1} \in \mathbb{IQ}_{>0}$, (that is, closed intervals
with positive rational endpoints) and \\
$\left(e_1,s_1\right), \ldots, \left(e_w,s_w\right)\in  \mathbb{Z}^2_{\geq
0} $ \\
such that
\begin{enumerate}
\item[] $\forall j \in \{1,\ldots,w-1\}, \;B$ has one and only one real root,
say $\sigma_j$, in $I_j$,
\item[] $\sigma_1 < \cdots <\sigma_{w-1}$, and
\item[] $\forall j \in \{1,\ldots,w\}\;\;\forall v\in\left(\sigma_{j-1},\sigma_j\right)
\;\; E_{v}=e_j\; \wedge S_{v}=s_j$
\end{enumerate}
where
\begin{enumerate}
\item[]$\sigma_0=0$, $\sigma_w=\infty$,
\item[]$E_{v}$ ($S_{v}$) denotes the number of (stable) equilibriums of 
$\dot{\mathbf x}={\mathbf f}\left(v,{\mathbf x}\right)$.
\end{enumerate}
\end{enumerate}
}
$B\leftarrow \text{\tt CriticalPolynomial}\left({\mathbf f}\right)$\;  {\nllabel{cp}}
 $I_1,\ldots,I_m\leftarrow$real root isolation of $B\left(\sigma\right)=0\wedge
\sigma>0$\;{\nllabel{isolate}}
 $v_1,\ldots,v_{m+1}\leftarrow$ rational points in each open interval of
$B\left(\sigma\right)\neq 0 \wedge \sigma>0$\;{\nllabel{sample}}
$Intervals\leftarrow$ empty list, $Numbers\leftarrow$ empty list\;{\nllabel{interval}}
\For{$j$  {\bf from} $1$ {\bf to} $m+1$ {\nllabel{loop2}}}
{
  $\left(e_j,s_j\right)\leftarrow \text{\tt EquilibriumCounting}
     \left({\mathbf f},v_j\right)$\;{\nllabel{ECv}}
  \If{$j>1${\nllabel{begincheck}}}
  {
   \If{$e_{j}=e_{j-1}$ {\bf and} $s_{j}=s_{j-1}$
          {\nllabel{checkcondition}}}
    {
      $e\leftarrow$number of the equilibriums  when $B\left(\sigma\right)=0$
          and $\sigma\in I_{j}$\;{\nllabel{check1}}
       $s\leftarrow$number of the stable equilibriums  when
        $B\left(\sigma\right)=0$ and $\sigma\in I_{j}$\;{\nllabel{check2}}
     \If{$e=e_{j}$ {\bf and} $s=s_{j}$}
            {{\bf next}\;}
     }
   $Intervals\leftarrow$ Append $I_{j-1}$ to  $Intervals$\; {\nllabel{endcheck}}
 }
   $Numbers\leftarrow$ Append $\left(e_{j},s_{j}\right)$ to $Numbers$\;
{\nllabel{endloop}}
}
\Return $B, Intervals,Numbers$\;
\end{algorithm}

\begin{algorithm}[hb]\label{CP}
\caption{\tt CriticalPolynomial}
\KwIn{
\begin{enumerate}
\item[] ${\mathbf f}=\left(f_1,\ldots,f_n\right)\in \left({\mathbb Q}\left(\sigma,{\mathbf
x}\right)\right)^n$ such
that $\dot{\mathbf x}={\mathbf f}$ is a MSRS

\end{enumerate}}
\KwOut{
\begin{enumerate}
\item[] $B\in \mathbb{Z}[\sigma]$ such that if $v$ is critical for $MSRS\left(l,g,h,p,\sigma\right)$,
then $B\left(v\right)=0$
\end{enumerate}}
$F,G_1,G_2\leftarrow \text{\tt DiagonalEquilibrium}\left({\mathbf f}\right)$\;
\nllabel{DECP1}
Compute $B_0$ such that
 $\big[\exists q(F=0\wedge
 G_1G_2=0)\big]\Rightarrow B_0\left(\sigma\right)=0$\; \nllabel{DECP2}
 \For{$i$ {\bf from} $1$ {\bf to} $\lfloor\frac{n}{2}\rfloor$\nllabel{CPloop}}
{$F_1,F_2,G_1,G_2,G_3,G_4\leftarrow \text{\tt NonDiagonalEquilibrium}
\left({\mathbf f},i\right)$\;\nllabel{NDECP1}
Compute $B_i$ such that $\big[\exists p,q(F_1=0\wedge
F_2=0\wedge G_1G_2G_3G_4=0)\big]\Rightarrow
B_i\left(\sigma\right)=0$\;{\nllabel{NDECP2}}}
$B\leftarrow \prod_{i=0}^{\lfloor\frac{n}{2}\rfloor}B_i$\;{\nllabel{CPF}}
{\bf return} $B$\;
\end{algorithm}
\begin{algorithm}\label{EC}
\caption{\tt EquilibriumCounting}
\KwIn{
\begin{enumerate}
\item[] ${\mathbf f}=\left(f_1,\ldots,f_n\right)\in \left({\mathbb Q}\left(\sigma,{\mathbf
x}\right)\right)^n$ such
that $\dot{\mathbf x}={\mathbf f}$ is a MSRS
\item[] $v$, a positive real number
\end{enumerate}}
\KwOut{
\begin{enumerate}
\item[]$\left(e, s\right)$ such that $E_v=e\wedge S_v=s$,
where
$E_{v}$ ($S_{v}$) denotes the number of (stable) equilibrium
of  $\dot{\mathbf x}={\mathbf f}\left(v,{\mathbf x}\right)$.

\end{enumerate}}

$F,G_1,G_2\leftarrow
\text{\tt DiagonalEquilibrium}\left({\mathbf f}\right)$\;{\nllabel{ECDE}}
$e\leftarrow$number of   positive roots of $\sigma=v\wedge F=0$\;
   {\nllabel{nde}}
  $s\leftarrow$ number of  positive roots of
   $\sigma=v\wedge F=0\wedge G_1<0 \wedge G_2<0$\; {\nllabel{nsde}}
\For{$i$ {\bf from} $1$ {\bf to} $\lfloor\frac{n}{2}\rfloor$ {\nllabel{subloop}}}
  {$F_1,F_2,G_1,
  G_2,G_3,G_4\leftarrow
\text{\tt NonDiagonalEquilibrium}\left({\mathbf f},i\right)$\;
{\nllabel{ECNDE}}
   $\tilde{e}\leftarrow$ number of   positive solutions of $\sigma=v\wedge
F_1=0\wedge
F_2=0\wedge p\neq q$\;{\nllabel{lnnd}}
  \eIf{$i= 1${\nllabel{ns1}}}
   {\eIf{$n=2$}{$\tilde{s}\leftarrow$number of  positive solutions of
   $\sigma=v\wedge F_1=0\wedge F_2=0\wedge p\neq q
\wedge G_3<0\wedge G_4>0${\nllabel{nscondition0}}}{$\tilde{s}\leftarrow$number
of  positive solutions of
   $\sigma=v\wedge F_1=0\wedge F_2=0\wedge p\neq q\wedge G_1<0
\wedge G_3<0\wedge G_4>0$ \;{\nllabel{nscondition1}}}}
{ $\tilde{s}\leftarrow$number of   positive solutions of
 $\sigma=v\wedge F_1=0\wedge F_2=0\wedge p\neq q\wedge G_1<0 \wedge G_2<0\wedge
G_3<0\wedge G_4>0$\;
{\nllabel{nscondition2}}}

 \eIf{$i=\frac{n}{2}$ \nllabel{checksymmetry}}
 {$e\leftarrow e+\frac{\tilde{e}}{2}\cdot \binom{n}{i}$,
  $s\leftarrow s+\frac{\tilde{s}}{2}\cdot \binom{n}{i}$\;}
 {$e\leftarrow e+\tilde{e}\cdot \binom{n}{i}$,
 $s\leftarrow s+\tilde{s}\cdot \binom{n}{i}$\;
 {\nllabel{compute2}}
  }
}
{\bf return} $\left(e,s\right)$\;
\end{algorithm}
\begin{algorithm}\label{DE}
\caption{\tt DiagonalEquilibrium}
\KwIn{
\begin{enumerate}
\item[] ${\mathbf f}=\left(f_1,\ldots,f_n\right)\in \left({\mathbb Q}\left(\sigma,{\mathbf
x}\right)\right)^n$ such
that $\dot{\mathbf x}={\mathbf f}$ is a MSRS
\end{enumerate}}
\KwOut{
\begin{enumerate}
\item[] $F,G_1,G_2 \in \mathbb{Q}\left(\sigma,q\right)$ such that for every
$\sigma\in {\mathbb R}_{>0}$,
\begin{enumerate}
\item[(1)] ${\mathbf r}=\left(q,\ldots,q\right)$ is an equilibrium  if and
only
if $F=0$

\item[(2)] if ${\mathbf r}=\left(q,\ldots,q\right)$ is an equilibrium then
the
eigenvalues
of $J_{\mathbf f}\left({\mathbf r}\right)$ are
\[\lambda _{1}=\cdots =\lambda _{n-1}=G_1,\lambda _{n}=G_2 \]
\end{enumerate}
\end{enumerate}}

Let $l,g,h,P$ be the functions such that $f_{k}=-l\left(x_k\right)+\sigma
\frac{g\left(x_k\right)}{P\left(x_1,\ldots,x_n\right)+h\left(x_k\right)}$\;
$D_n\leftarrow-\frac{P\left(x_1,\ldots,x_n\right)+h\left(x_n\right)}{l\left(x_n\right)}$\;
$
\tau   \leftarrow \frac{\partial f_n}{\partial x_n},\;\;\;\;
\xi    \leftarrow \frac{\frac{\partial P}{\partial x_{n-1}}}{D_n}
,\;\;\;\;
F\leftarrow f_1$\;
$G_1\leftarrow \tau-\xi,\;\;\;\;
G_2\leftarrow {\tau}+\left(n-1\right)\xi$\;
Replace $x_1,\ldots,x_n$ with $q$ in $F,G_1,G_2$\;
{\bf return} $F,G_1,G_2$\;
\end{algorithm}

\begin{algorithm}\label{NDE}
\caption{\tt NonDiagonalEquilibrium}
\KwIn{
\begin{enumerate}
\item[] ${\mathbf f}=\left(f_1,\ldots,f_n\right)\in \left({\mathbb Q}\left(\sigma,{\mathbf
x}\right)\right)^n$ such that $\dot{\mathbf x}={\mathbf f}$ is a MSRS 
\item[] $i$, an positive integer such that $1\leq i\leq \lfloor\frac{n}{2}\rfloor$
\end{enumerate}}
\KwOut{
\begin{enumerate}
\item[] $F_1, F_2, G_1, G_2, G_3, G_4 \in \mathbb{Q}\left(\sigma,p,q\right)$
such that for every
$\sigma\in {\mathbb R}_{>0}$,
\begin{enumerate}
\item[(1)] ${\mathbf r}=\left(p,\ldots,p,q,\ldots,q\right)$ is an equilibrium
and $p$ appears
$i$ times
if and only if  $F_1=0\wedge F_2=0$
\item[(2)] If ${\mathbf r}=\left(p,\ldots,p,q,\ldots,q\right)$ is an equilibrium
and $p$ appears
$i$ times then
the  eigenvalues
of $J_{\mathbf f}\left({\mathbf r}\right)$ are
as follows.\begin{enumerate}
\item[(a)] if $i=1$, then $$\lambda_1=\cdots=\lambda_{n-2}=G_1, \; \lambda_{n-1}+\lambda_n=G_3,
\; \lambda_{n-1}\lambda_n=G_4$$
\item[(b)]  if $i>1$, then $$\begin{array}{l}\lambda_1=\cdots=\lambda_{n-i-1}=G_1,
\;
\lambda_{n-i}=\cdots=\lambda_{n-2}=G_{2}, \\
 \lambda_{n-1}+\lambda_n=G_3, \;
\lambda_{n-1}\lambda_n=G_4
\end{array}$$
\end{enumerate}
\end{enumerate}
\end{enumerate}}

Let $l,g,h,P$ be the functions such that  $f_{k}=-l\left(x_k\right)+\sigma
\frac{g\left(x_k\right)}{P\left(x_1,\ldots,x_n\right)+h\left(x_k\right)}$\;
$D_k\leftarrow\frac{P\left(x_1,\ldots,x_n\right)+h\left(x_k\right)}{l\left(x_k\right)}$
for $k=1,n$\;
$
\beta  \leftarrow \frac{\partial f_1}{\partial x_1},\;\;\;\;
\tau   \leftarrow \frac{\partial f_n}{\partial x_n},\;\;\;\;
\gamma \leftarrow \frac{\frac{\partial P}{\partial x_2}}{D_1},\;\;\;\;
\xi    \leftarrow \frac{\frac{\partial P}{\partial x_{n-1}}}{D_n},\;\;\;\;
\mu    \leftarrow \frac{\frac{\partial P}{\partial x_n}}{D_1},\;\;\;\;
\nu    \leftarrow \frac{\frac{\partial P}{\partial x_1}}{D_n}$\;

$F1\leftarrow f_1,\;\;\;\;F2\leftarrow f_n,\;\;\;\;G_1\leftarrow \tau-\xi,\;\;\;\;
G_2\leftarrow \beta-\gamma$\;
$G_3\leftarrow \beta+\tau+\left(i-1\right)\gamma+\left(n-i-1\right)\xi$\;
$G_4\leftarrow \left(\beta+\left(i-1\right)\gamma\right)\left(\tau+\left(n-i-1\right)\xi\right)-i\left(n-i\right)\mu\nu$\;
Replace $x_1,\ldots,x_i$ with $p$ and $x_{i+1},\ldots,x_n$ with $q$ in $F_1,F_2,G_1,G_2,G_3,G_4$\;

{\bf return} $F_1,F_2,G_1,G_2,G_3,G_4$\;

\end{algorithm}

\begin{example}\label{example}
We will illustrate the algorithm on
Example \ref{result}.


\item[]In {\bf Algorithm \ref{sa}. Line \ref{cp}},   we compute $B\left(\sigma\right)$
by Algorithm \ref{CP}.
\begin{enumerate}
\item[]In {\bf Algorithm \ref{CP}. Line \ref{DECP1}}, we call $\text{\tt
DiagonalEquilibrium}\left(f_1,f_2,f_3,f_4\right)$, where
\[f_k=-x_k+\frac{\sigma}
{1+x_1^4+x_2^4+x_3^4+x_4^4-x_k^4},\;\;k=1,\ldots,4,\]
and get
\begin{align*}
\begin{cases}
F\left(\sigma,  q\right)=- q + \frac{\sigma}{1+3q^4} 
\\
G_1\left(\sigma,q\right)=-1+\frac{4q^4}{1+3q^4} 
\\
G_2\left(\sigma,q\right)=-1-\frac{12q^4}{1+3q^4 }.
\end{cases}
\end{align*}
\item[]In {\bf Algorithm \ref{CP}. Line \ref{DECP2}}, we compute the projection
of $F=0\wedge G_1G_2=0$ on $\sigma$ axe and obtain $B_0\left(\sigma\right)=\sigma-4.$

\item[]In {\bf Algorithm \ref{CP}. Line \ref{CPloop}}, we start loop.  Note
that $\lfloor \frac{n}{2}\rfloor=2$, so $i=1,2$.
\begin{enumerate}
\item[]For $i=1$, in {\bf Algorithm \ref{CP}. Line \ref{NDECP1}}, we call
\[\text{\tt
NonDiagonalEquilibrium}\left(f_{1},f_2,f_3,f_4,1\right)\] and get
  \begin{align*}
\begin{cases}
F_1\left(\sigma,p,q\right)=- p + \frac{\sigma}{1  + 3q^{4}}  \\
F_{2}\left(\sigma,p,q\right)=- q + \frac{ \sigma}{1 + p^{4} + 2q^{4}} \\
G_{1}\left(\sigma,p,q\right)=-1 + \frac{4q^4}{1 + p^{4} + 2q^{4} }\\
G_{2}\left(\sigma,p,q\right)=-1 + \frac{4q^3p}{1 + 3q^4 }\\
G_{3}\left(\sigma,p,q\right)=-2 - \frac{8q^4}{1 + p^{4} + 2q^{4} }\\
G_{4}\left(\sigma,p,q\right)=1+\frac{8q^4}{1 + p^{4} + 2q^{4} }-\frac{48q^4p^4}{(1
+ 3q^4)(1 + p^{4} + 2q^{4})}
\\
\end{cases}.
\end{align*}
Then in {\bf Algorithm \ref{CP}. Line \ref{NDECP2}}, we compute the projection
of $F_1=0\wedge F_2=0\wedge G_1G_2G_3G_4=0$ on $\sigma$ axe and obtain
\[B_1=(\sigma-4)
(42755090541778564453125\sigma^{24}+\cdots-140737488355328).\]

\item[]For
$i=2$, in {\bf Algorithm \ref{CP}. Line \ref{NDECP1}}, we call
\[\text{\tt
NonDiagonalEquilibrium}\left(f_1,f_2,f_3,f_4,2\right)\] and get
\begin{align*}
\begin{cases}
F_1\left(\sigma,p,q\right)=- p + \frac{\sigma}{1  + p^4+2q^{4}} \\
F_{2}\left(\sigma,p,q\right)=- q + \frac{\sigma}{1 + 2p^{4} + q^{4}}  \\
G_{1}(\sigma,p,p)=-1 + \frac{4q^4}{1 + 2p^{4} + q^{4} } \\
G_{2}\left(\sigma,p,q\right)=-1 + \frac{4p^4}{1  + p^4+2q^{4}}\\
G_{3}\left(\sigma,p,q\right)=-2 - \frac{4p^4}{1  + p^4+2q^{4}}  - \frac{4q^4}{1
+ 2p^{4} + q^{4} }\\
G_4\left(\sigma,p,q\right)= \left(-1-\frac{4p^4}{1  + p^4+2q^{4}}\right)\left(-1-\frac{4q^4}{1
+ 2p^{4} + q^{4} }\right)-\frac{64q^4p^4}{\left(1  + p^4+2q^{4}\right)\left(1
+ 2p^{4} + q^{4} \right)}
\end{cases}.
\end{align*}
Then in {\bf Algorithm \ref{CP}. Line \ref{NDECP2}}, we compute the projection
of $F_1=0\wedge F_2=0\wedge G_1G_2G_3G_4=0$ on $\sigma$ axe and obtain
\[B_2=\sigma-4.\]
\end{enumerate}
\item[]In  {\bf Algorithm \ref{CP}. Line \ref{CPF}}, let $B=B_0B_1B_2$.
\end{enumerate}
\item[]In {\bf Algorithm \ref{sa}. Line \ref{isolate}}, we isolate the  positive
roots
of $B\left(\sigma\right)=0$, obtaining
\[I_1=[\frac{5}{4}, \frac{21}{16}], I_2=[4, 4].\]
\item[]In {\bf Algorithm \ref{sa}. Line \ref{sample}}, sample rational points
from $\left(0, \frac{5}{4}\right)$, $\left(\frac{21}{16},4\right)$, and $\left(4,\infty\right)$,
obtaining
\[v_1=1, v_2=2, v_3=5.\]
\item[]In {\bf Algorithm \ref{sa}. Line \ref{loop2}}, we start the  loop
and compute the number of (stable) equilibriums for very sample point.
\begin{enumerate}
\item[]
For $j=1$, in {\bf Algorithm \ref{sa}. Line \ref{ECv}}, call \text{\tt
EquilibriumCounting}$(f_1,f_2,f_3,f_4,1)$.

\begin{enumerate}
\item[]
In {\bf Algorithm \ref{EC}. Lines \ref{ECDE}--\ref{nsde}}, compute the number
of (stable) diagonal equilibriums and
initialize $e_1=1$ ($s_1=1$).
\item[]In  {\bf Algorithm \ref{EC}. Line \ref{subloop}}, we enter the loop.
\begin{enumerate}
\item[] For $i=1$,  in {\bf Algorithm \ref{EC}. Lines \ref{ECNDE}--\ref{lnnd}},
 compute the number of positive solutions of
\[\sigma=1\wedge F_1=0\wedge
F_2=0\wedge p\neq q,\]
obtaining $0$.

\item[]For $i=2$, in {\bf Algorithm \ref{EC}.  Lines \ref{ECNDE}--Line \ref{lnnd}},
 compute the number of positive solutions of
\[\sigma=1\wedge F_1=0\wedge
F_2=0\wedge p\neq q,\] obtaining $0$.
\end{enumerate}

\end{enumerate}
So when $\sigma=1$,  there
is only $1$ equilibrium, that is the diagonal one,  and it is stable.
\item[]Note we do not pass through {\bf Algorithm \ref{sa}. Lines \ref{checkcondition}--\ref{endcheck}}.
\item[]In {\bf Algorithm \ref{sa}. Line \ref{endloop}}, let $Numbers=[\left(1,1\right)]$.

\item[]For $j=2$, call \text{\tt
EquilibriumCounting}$\left(f_1,f_2,f_3,f_4,2\right)$.

\begin{enumerate}
\item[]In {\bf Algorithm \ref{EC}. Lines \ref{ECDE}--\ref{nsde}}, compute
the number
of (stable) diagonal equilibriums and
initialize  $e_2=1$ ($s_2=1$).
\item[]In  {\bf Algorithm \ref{EC}. Line \ref{subloop}}, we enter the loop.
\begin{enumerate}

\item[]For $i=1$,  in {\bf Algorithm \ref{EC}. Lines \ref{ECNDE}--\ref{lnnd}},
 compute the number of positive solutions of
\[\sigma=2\wedge F_1=0\wedge
F_2=0\wedge p\neq q,\]obtaining $2$.  Then in{\bf \ Algorithm \ref{EC}. Lines
\ref{nscondition2}}, compute the number of  distinct
positive solutions of
\[\sigma=2\wedge F_1=0\wedge
F_2=0\wedge p\neq q\wedge G_1<0\wedge G_3<0\wedge G_4>0,\]
obtaining $1$.

\item[] For $i=2$, in {\bf Algorithm \ref{EC}. Lines \ref{ECNDE}--\ref{lnnd}},
 compute the number of positive solutions of
\[\sigma=2\wedge F_1=0\wedge
F_2=0\wedge p\neq q,\]  obtaining $0$.
\end{enumerate}
\item[]In {\bf Algorithm \ref{EC}. Lines \ref{checksymmetry}--\ref{compute2}},
let $e_2=1+2\cdot\binom{4}{1}=9$ and $s_2=1+\binom{4}{1}=5$.

\end{enumerate}
So when $\sigma=2$,  there are $9$ equilibriums and $5$ stable equilibriums.

\item[] Since $e_1\neq e_2$,
in {\bf Algorithm \ref{sa}. Lines \ref{checkcondition}}, {\bf \ref{endcheck}}
and {\bf \ref{endloop}},
let $Intervals=[I_1]$ and let $Numbers=[\left(1,1\right), \left(9,5\right)]$.

\item[]For $j=3$, call \text{\tt
EquilibriumCounting}$\left(f_1,f_2,f_3,f_4, 5\right)$.

\begin{enumerate}
\item[]In {\bf Algorithm \ref{EC}. Lines \ref{ECDE}--\ref{nsde}}, compute
the number
of (stable) diagonal equilibriums and
initialize  $e_3=1$ ($s_3=0$).
\item[]In  {\bf Algorithm \ref{EC}. Line \ref{subloop}}, we enter the loop.
\begin{enumerate}
\item[]For $i=1$, in {\bf Algorithm \ref{EC}. Lines \ref{ECNDE}--\ref{lnnd}},
 compute the number of positive solutions of
\[\sigma=5\wedge F_1=0\wedge
F_2=0\wedge p\neq q,\]obtaining $2$.   Then in {\bf Algorithm \ref{EC}. Lines
\ref{nscondition2}}, compute the number of  distinct
positive solutions of
\[\sigma=5\wedge F_1=0\wedge
F_2=0\wedge p\neq q\wedge G_1<0\wedge G_3<0\wedge G_4>0,\]
obtaining $1$.
\item[]For $i=2$, in {\bf Algorithm \ref{EC}. Lines \ref{ECNDE}--\ref{lnnd}},
 compute the number of positive solutions of
\[\sigma=5\wedge F_1=0\wedge
F_2=0\wedge p\neq q,\] obtaining $2$. Then in {\bf Algorithm \ref{EC}. Lines
\ref{nscondition1}}, compute the number of  distinct
positive solutions of
\[\sigma=5\wedge F_1=0\wedge
F_2=0\wedge p\neq q\wedge G_{1}<0\wedge G_2<0\wedge G_3<0\wedge G_4>0,\]
obtaining $0$.
\end{enumerate}
In {\bf Algorithm \ref{EC}. Lines \ref{checksymmetry}--\ref{compute2}}, let
$e_3=1+2\cdot\binom{4}{1}+\frac{2\cdot\binom{4}{2}}{2}=15$ and $s_3=\binom{4}{1}=4$.

\end{enumerate}
So when $\sigma=5$,  there are $15$ equilibriums and $4$ stable equilibriums.

\item[] Since $e_2\neq e_3$,
in {\bf Algorithm \ref{sa}.  Lines \ref{checkcondition}}, {\bf \ref{endcheck}}
and {\bf \ref{endloop}},
let $Intervals=[I_1,I_2]$ and let $Numbers=[\left(1,1\right),
\left(9,5\right),\left(15,4\right)]$.

\end{enumerate}
Finally, the main algorithm outputs  shown in Example \ref{result}.
\end{example}

\section{Performance}\label{performance}

\begin{figure}[tp]
\begin{center}
\tiny
\begin{tabular}{|c|c|c|c|c|c|c|c|c|c|c|c|c|c|c|c|}\hline
\backslashbox{$n$}{$c$}  &  $  1$  &  $  2$  &  $  3$  &  $  4$  &  $  5$  &  $  6$  &  $  7$  &  $  8$  &  $  9$  &  $ 10$  &  $ 11$  &  $ 12$  &  $ 13$  &  $ 14$  &  $ 15$\\\hline
{\multirow{2}{*}{$  2$}}  &  $       0.0$  &  $       0.0$  &  $       0.1$  &  $       0.1$  &  $       0.1$  &  $       0.2$  &  $       0.2$  &  $       0.4$  &  $       0.7$  &  $       1.4$  &  $       2.2$  &  $       3.6$  &  $       5.8$  &  $       9.1$  &  $      13.8$\\
  &  $       0.0$  &  $       0.0$  &  $       0.0$  &  $       0.0$  &  $       0.1$  &  $       0.1$  &  $       0.1$  &  $       0.2$  &  $       0.3$  &  $       0.7$  &  $       1.1$  &  $       1.6$  &  $       2.6$  &  $       4.2$  &  $       5.9$\\\hline
{\multirow{2}{*}{$  3$}}  &  $       0.0$  &  $       0.0$  &  $       0.1$  &  $       0.1$  &  $       0.2$  &  $       0.3$  &  $       0.5$  &  $       0.9$  &  $       1.9$  &  $       3.2$  &  $       6.3$  &  $      10.4$  &  $      19.4$  &  $      29.3$  &  $      53.9$\\
  &  $       0.1$  &  $       1.7$  &  $      96.9$  &  $    \infty$  &  $    \infty$  &  $    \infty$  &  $    \infty$  &  $    \infty$  &  $    \infty$  &  $    \infty$  &  $    \infty$  &  $    \infty$  &  $    \infty$  &  $    \infty$  &  $    \infty$\\\hline
{\multirow{2}{*}{$  4$}}  &  $       0.0$  &  $       0.1$  &  $       0.1$  &  $       0.2$  &  $       0.3$  &  $       0.7$  &  $       1.3$  &  $       2.4$  &  $       4.7$  &  $       9.1$  &  $      16.7$  &  $      28.6$  &  $      51.4$  &  $      85.5$  &  $     129.7$\\
  &  $       0.1$  &  $       3.1$  &  $    \infty$  &  $    \infty$  &  $    \infty$  &  $    \infty$  &  $    \infty$  &  $    \infty$  &  $    \infty$  &  $    \infty$  &  $    \infty$  &  $    \infty$  &  $    \infty$  &  $    \infty$  &  $    \infty$\\\hline
{\multirow{2}{*}{$  5$}}  &  $       0.0$  &  $       0.1$  &  $       0.1$  &  $       0.2$  &  $       0.4$  &  $       0.7$  &  $       1.6$  &  $       2.9$  &  $       6.2$  &  $      11.5$  &  $      22.5$  &  $      36.7$  &  $      67.8$  &  $     110.6$  &  $     192.4$\\
  &  $       0.2$  &  $       0.1$  &  $    \infty$  &  $    \infty$  &  $    \infty$  &  $    \infty$  &  $    \infty$  &  $    \infty$  &  $    \infty$  &  $    \infty$  &  $    \infty$  &  $    \infty$  &  $    \infty$  &  $    \infty$  &  $    \infty$\\\hline
{\multirow{2}{*}{$  6$}}  &  $       0.1$  &  $       0.1$  &  $       0.1$  &  $       0.2$  &  $       0.6$  &  $       1.3$  &  $       2.6$  &  $       5.0$  &  $      10.1$  &  $      18.8$  &  $      36.2$  &  $      65.6$  &  $     111.9$  &  $     192.1$  &  $     289.2$\\
  &  $       0.3$  &  $      16.7$  &  $    \infty$  &  $    \infty$  &  $    \infty$  &  $    \infty$  &  $    \infty$  &  $    \infty$  &  $    \infty$  &  $    \infty$  &  $    \infty$  &  $    \infty$  &  $    \infty$  &  $    \infty$  &  $    \infty$\\\hline
{\multirow{2}{*}{$  7$}}  &  $       0.1$  &  $       0.1$  &  $       0.1$  &  $       0.3$  &  $       0.6$  &  $       1.3$  &  $       3.1$  &  $       5.7$  &  $      11.6$  &  $      22.0$  &  $      42.4$  &  $      70.9$  &  $     134.6$  &  $     220.4$  &  $     354.3$\\
  &  $       0.1$  &  $     177.7$  &  $    \infty$  &  $    \infty$  &  $    \infty$  &  $    \infty$  &  $    \infty$  &  $    \infty$  &  $    \infty$  &  $    \infty$  &  $    \infty$  &  $    \infty$  &  $    \infty$  &  $    \infty$  &  $    \infty$\\\hline
{\multirow{2}{*}{$  8$}}  &  $       0.1$  &  $       0.1$  &  $       0.2$  &  $       0.3$  &  $       0.7$  &  $       1.7$  &  $       3.7$  &  $       8.3$  &  $      16.7$  &  $      31.9$  &  $      59.7$  &  $     107.2$  &  $     185.1$  &  $     296.9$  &  $     510.4$\\
  &  $    \infty$  &  $    \infty$  &  $    \infty$  &  $    \infty$  &  $    \infty$  &  $    \infty$  &  $    \infty$  &  $    \infty$  &  $    \infty$  &  $    \infty$  &  $    \infty$  &  $    \infty$  &  $    \infty$  &  $    \infty$  &  $    \infty$\\\hline
{\multirow{2}{*}{$  9$}}  &  $       0.1$  &  $       0.1$  &  $       0.2$  &  $       0.3$  &  $       0.9$  &  $       1.8$  &  $       4.2$  &  $       8.2$  &  $      18.8$  &  $      34.8$  &  $      67.0$  &  $     114.8$  &  $     213.6$  &  $     340.5$  &  $     590.3$\\
  &  $    \infty$  &  $    \infty$  &  $    \infty$  &  $    \infty$  &  $    \infty$  &  $    \infty$  &  $    \infty$  &  $    \infty$  &  $    \infty$  &  $    \infty$  &  $    \infty$  &  $    \infty$  &  $    \infty$  &  $    \infty$  &  $    \infty$\\\hline
{\multirow{2}{*}{$ 10$}}  &  $       0.1$  &  $       0.2$  &  $       0.2$  &  $       0.3$  &  $       0.9$  &  $       0.9$  &  $       1.8$  &  $      11.0$  &  $      21.6$  &  $      47.6$  &  $      88.8$  &  $     149.4$  &  $     266.8$  &  $     453.5$  &  $     703.0$\\
  &  $    \infty$  &  $    \infty$  &  $    \infty$  &  $    \infty$  &  $    \infty$  &  $    \infty$  &  $    \infty$  &  $    \infty$  &  $    \infty$  &  $    \infty$  &  $    \infty$  &  $    \infty$  &  $    \infty$  &  $    \infty$  &  $    \infty$\\\hline
{\multirow{2}{*}{$ 11$}}  &  $       0.1$  &  $       0.2$  &  $       0.2$  &  $       0.4$  &  $       0.8$  &  $       2.1$  &  $       5.5$  &  $      10.8$  &  $      23.9$  &  $      43.9$  &  $      94.2$  &  $     161.8$  &  $     293.6$  &  $     482.8$  &  $     768.1$\\
  &  $    \infty$  &  $    \infty$  &  $    \infty$  &  $    \infty$  &  $    \infty$  &  $    \infty$  &  $    \infty$  &  $    \infty$  &  $    \infty$  &  $    \infty$  &  $    \infty$  &  $    \infty$  &  $    \infty$  &  $    \infty$  &  $    \infty$\\\hline
{\multirow{2}{*}{$ 12$}}  &  $       0.1$  &  $       0.2$  &  $       0.3$  &  $       0.4$  &  $       1.0$  &  $       2.3$  &  $       6.7$  &  $      13.6$  &  $      29.1$  &  $      58.2$  &  $     102.0$  &  $     204.7$  &  $     359.5$  &  $     604.9$  &  $    1029.0$\\
  &  $    \infty$  &  $    \infty$  &  $    \infty$  &  $    \infty$  &  $    \infty$  &  $    \infty$  &  $    \infty$  &  $    \infty$  &  $    \infty$  &  $    \infty$  &  $    \infty$  &  $    \infty$  &  $    \infty$  &  $    \infty$  &  $    \infty$\\\hline
{\multirow{2}{*}{$ 13$}}  &  $       0.1$  &  $       0.2$  &  $       0.4$  &  $       0.4$  &  $       1.0$  &  $       2.5$  &  $       6.7$  &  $      15.1$  &  $      33.5$  &  $      67.6$  &  $     133.6$  &  $     207.0$  &  $     414.7$  &  $     662.5$  &  $    1078.1$\\
  &  $    \infty$  &  $    \infty$  &  $    \infty$  &  $    \infty$  &  $    \infty$  &  $    \infty$  &  $    \infty$  &  $    \infty$  &  $    \infty$  &  $    \infty$  &  $    \infty$  &  $    \infty$  &  $    \infty$  &  $    \infty$  &  $    \infty$\\\hline
{\multirow{2}{*}{$ 14$}}  &  $       0.1$  &  $       0.2$  &  $       0.3$  &  $       0.6$  &  $       1.1$  &  $       2.6$  &  $       7.0$  &  $      15.9$  &  $      37.1$  &  $      74.8$  &  $     143.4$  &  $     259.6$  &  $     415.7$  &  $     812.0$  &  $    1319.1$\\
  &  $    \infty$  &  $    \infty$  &  $    \infty$  &  $    \infty$  &  $    \infty$  &  $    \infty$  &  $    \infty$  &  $    \infty$  &  $    \infty$  &  $    \infty$  &  $    \infty$  &  $    \infty$  &  $    \infty$  &  $    \infty$  &  $    \infty$\\\hline
{\multirow{2}{*}{$ 15$}}  &  $       0.1$  &  $       0.2$  &  $       0.3$  &  $       0.6$  &  $       1.1$  &  $       2.7$  &  $       7.0$  &  $      16.4$  &  $      39.4$  &  $      78.3$  &  $     151.1$  &  $     274.8$  &  $     501.3$  &  $     731.3$  &  $    1427.3$\\
  &  $    \infty$  &  $    \infty$  &  $    \infty$  &  $    \infty$  &  $    \infty$  &  $    \infty$  &  $    \infty$  &  $    \infty$  &  $    \infty$  &  $    \infty$  &  $    \infty$  &  $    \infty$  &  $    \infty$  &  $    \infty$  &  $    \infty$\\\hline
\end{tabular}

\end{center}
\caption{Timings  of the special algorithm (Algorithm~\ref{sa}) and the general
algorithm }\label{timings}
\bigskip
\begin{center}
\includegraphics[height=3in]{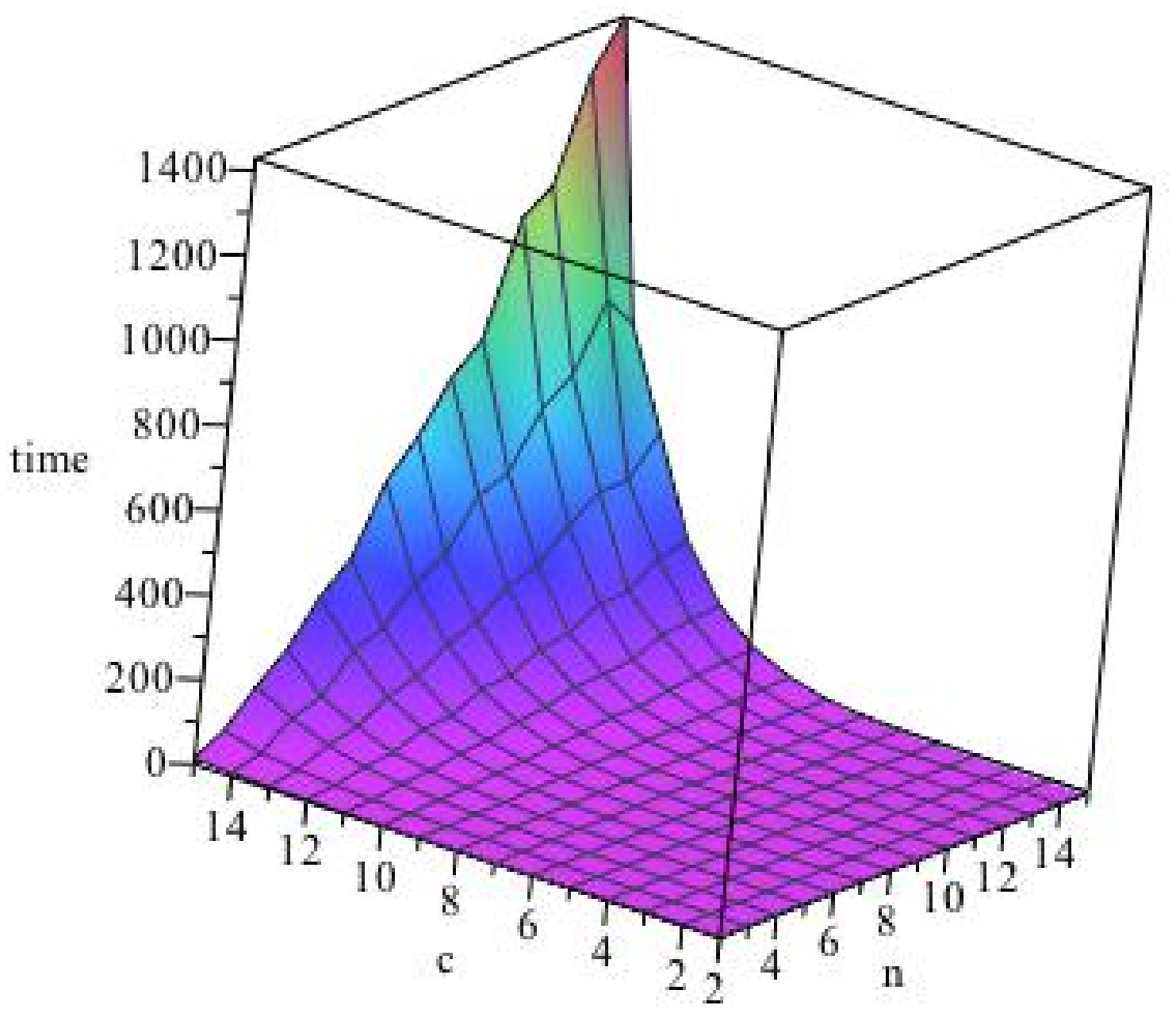}
\end{center}
\caption{$time$ -- $\left(n,c\right)$ of Algorithm \ref{sa}}\label{time}
\end{figure}

\begin{figure}[tp]
\subfigure[$n=3$]{
\begin{pspicture*}(0,0)(6.5,5.5)
\psset{xunit=0.7cm, yunit=0.25cm}
\psaxes{->}(0,0)(8,20)
\pscurve[linecolor=red]
(1.466,20)(1.500000000,18.05768355)(1.666666667,9.665910225)(2.,3.744480221)(2.166666667,2.932798584)
(2.333333333,2.437350705)(2.500000000,2.110124257)
(2.666666667,1.881180152)(3.,1.587270600)(3.333333333,1.410689094)
(3.500000000,1.347371056)(3.666666667,1.295297690)(4.,1.215318693)
(4.333333333,1.157430378)(4.500000000,1.134286871)(4.666666667,1.114123395)
(5.,1.080873798)(5.333333333,1.054809369)(5.500000000,1.043850652)
(5.666666667,1.034027547)(6.,1.017223224)(6.333333333,1.003474503)
(6.500000000,.9975298610)(6.666666667,.9921137624)(7.,.9826470661)
(8.,.9623001836)
\pscurve[linecolor=blue](2.22,20)(2.333333333,11.20930220)
(2.500000000,6.597539554)(2.666666667,4.656871062)(3.,3.)(3.333333333,2.293286887)
(3.500000000,2.078093084)(3.666666667,1.913896624)(4.,1.681792831)
(4.333333333,1.527311397)(4.500000000,1.468389912)(4.666666667,1.418270386)
(5.,1.337902603)(5.333333333,1.276674600)(5.500000000,1.251333820)
(5.666666667,1.228793555)(6.,1.190550789)(6.333333333,1.159469666)
(6.500000000,1.146058034)(6.666666667,1.133837830)(7.,1.112436367)
(8.,1.065785556)
\rput(8.5,0.5){$c$}
\rput(0.4,20){$\sigma$}
 \rput(1.5,20.5){$E_1$}
 \rput(2.3,20.5){$E_2$}
\rput(1.2,3.7){\tiny$\left(1,1\right)$}
\rput(2,10.5){\tiny$\left(7,4\right)$}
\rput(5.3,9.5){\tiny$\left(7,3\right)$}
 \end{pspicture*}}
\subfigure[$n=4$]{
\begin{pspicture*}(0,0)(6.5,5.5)
\psset{xunit=0.7cm, yunit=0.25cm}
\psaxes{->}(0,0)(8,20)
\pscurve[linecolor=red]
(1.556,20)(1.666666667,12.93605411)(2.,4.856380818)(2.333333333,2.929993587)
(2.500000000,2.475736091)(2.666666667,2.165598299)(3.,1.777446336)
(3.333333333,1.549955990)(3.500000000,1.469453549)(3.666666667,1.403666207)
(4.,1.303331342)(4.333333333,1.231195790)(4.500000000,1.202452016)
(4.666666667,1.177444366)(5.,1.136253889)(5.333333333,1.103972260)
(5.500000000,1.090389846)(5.666666667,1.078204201)(6.,1.057319880)(6.333333333,1.040175681)
(6.500000000,1.032739353)(6.666666667,1.025947964)
(7.,1.014030122)(8.,.9880894127)
\pscurve[linecolor=blue](3.256,20)(3.333333333,13.90389170)(3.500000000,8.533095579)
(3.666666667,6.143099906)(4.,4.)(4.333333333,3.041244558)(4.500000000,2.741510186)
(4.666666667,2.509683836)(5.,2.176376408)(5.333333333,1.949962661)(5.500000000,1.862387088)
(5.666666667,1.787261012)(6.,1.665366355)(6.333333333,1.571063718)
(6.500000000,1.531590050)(6.666666667,1.496223928)(7.,1.435586873)(8.,1.308424694)
\rput(8.5,0.5){$c$}
\rput(0.4,20){$\sigma$}
 \rput(1.6,20.5){$E_1$}
 \rput(3.5,20.5){$E_2$}
\rput(1.4,3.5){\tiny$\left(1,1\right)$}
\rput(2.5,11.5){\tiny$\left(9,5\right)$}
\rput(5.3,10.5){\tiny$\left(15,4\right)$}
 \end{pspicture*}}
\subfigure[$n=5$]{
\begin{pspicture*}(0,0)(6.5,5.5)
\psset{xunit=0.7cm, yunit=0.25cm}
\psaxes{->}(0,0)(8,20)
\pscurve[linecolor=red]
(1.62,20)(1.666666667,17.13618679)(2.,5.753282311)(2.333333333,3.297289806)(2.500000000,2.741237120)
(2.666666667,2.367818807)(3.,1.908420862)(3.333333333,1.643695348)
(3.500000000,1.550871559)(3.666666667,1.475354856)(4.,1.360768448)
(4.333333333,1.278810896)(4.500000000,1.246245046)(4.666666667,1.217950485)
(5.,1.171413064)(5.333333333,1.134983728)(5.500000000,1.119661922)(5.666666667,1.105915900)
(6.,1.082350640)(6.333333333,1.062989328)(6.500000000,1.054583235)(6.666666667,1.046900140)
(7.,1.033398914)(8.,1.003874850)
\pscurve[linecolor=blue](4.1,20)(4.333333333,9.291743603)
(4.500000000,6.896492427)(4.666666667,5.509178890)(5.,3.992231088)
(5.333333333,3.193133784)(5.500000000,2.922977914)(5.666666667,2.706231428)
(6.,2.381253453)(6.333333333,2.150356937)(6.500000000,2.058578178)(6.666666667,1.978654786)
(7.,1.846462685)(8.,1.587599363)
\pscurve[linecolor=black](4.285,20)(4.333333333,16.75129000)(4.500000000,10.49876136)
(4.666666667,7.635401540)(5.,5.)(5.333333333,3.789954345)(5.500000000,3.406079949)
(5.666666667,3.106913324)(6.,2.672696154)(6.333333333,2.374400435)(6.500000000,2.258144910)
(6.666666667,2.157966502)(7.,1.994419933)(8.,1.681792831)
\rput(8.5,0.5){$c$}
\rput(0.4,20){$\sigma$}
 \rput(1.6,20.5){$E_1$}
 \rput(3.8,20.5){$E_2$}
\rput(4.6,20.5){$E_3$}
\rput(1.5,3.5){\tiny$\left(1,1\right)$}
\rput(3.3,6.5){\tiny$\left(11,6\right)$}
\rput(4.4,12.5){\tiny$\left(31,6\right)$}
\rput(6.3,9.5){\tiny$\left(31,5\right)$}
 \end{pspicture*}}
\subfigure[$n=6$]{
\begin{pspicture*}(0,0)(6.5,5.5)
\psset{xunit=0.7cm, yunit=0.25cm}
\psaxes{->}(0,0)(8,20)
\pscurve[linecolor=red](1.68,20)(1.750000000,14.36635692)(1.800000000,11.81021514)
(2.,6.526920830)(2.200000000,4.385794998)(2.250000000,4.051022139)(2.333333333,3.597064004)
(2.500000000,2.953990159)(2.666666667,2.527503835)(2.750000000,2.365647355)(3.,2.009592415)
(3.250000000,1.775589683)(3.333333333,1.714969410)(3.500000000,1.612385445)(3.666666667,1.529216920)
(3.750000000,1.493367984)(4.,1.403521897)(4.250000000,1.333807621)(4.333333333,1.313987732)(4.500000000,1.278493529)
(4.666666667,1.247690458)(4.750000000,1.233788113)(5.,1.197093307)(5.250000000,1.166575107)(5.333333333,1.157534305)
(5.500000000,1.140905560)(5.666666667,1.125990062)(6.,1.100422364)(6.333333333,1.079412045)(6.500000000,1.070286788)
(6.666666667,1.061943601)(7.,1.047273175)(8.,1.015117908)
\pscurve[linecolor=blue]
(4.65,20)(4.666666667,18.83229414)(4.750000000,14.35980715)(5.,8.341413292)(5.250000000,5.916678570)(5.333333333,5.408308635)
(5.500000000,4.634124038)(5.666666667,4.074641371)(6.,3.324829590)(6.333333333,2.849027733)
(6.500000000,2.671705041)(6.666666667,2.522415520)(7.,2.285484381)(8.,1.855059988)

\pscurve[linecolor=black](5.34,20)(5.333333333,19.65985081)(5.500000000,12.47743774)
(5.666666667,9.130485136)(6.,6.)(6.333333333,4.539065405)(6.500000000,4.071281527)(6.666666667,3.704952794)
(7.,3.170032825)(8.,2.324494781)
\rput(8.5,0.5){$c$}
\rput(0.4,20){$\sigma$}
 \rput(1.6,20.5){$E_1$}
 \rput(4.6,20.5){$E_2$}
 \rput(5.6,20.5){$E_3$}
\rput(1.5,3.5){\tiny$\left(1,1\right)$}
\rput(4,6.5){\tiny$\left(13,7\right)$}
\rput(5.3,9.5){\tiny$\left(43,7\right)$}
\rput(6.6,8.5){\tiny$\left(63,6\right)$}

\end{pspicture*}}
\subfigure[$n=7$]{
\begin{pspicture*}(0,0)(7.7,5.5)
\psset{xunit=0.55cm, yunit=0.25cm}
\psaxes{->}(0,0)(11,20)
\pscurve[linecolor=red]
(1.685,20)(2.,7.217631630)(2.333333333,3.853521730)
(2.500000000,3.133456137)(2.666666667,2.660700943)(3.,2.092575309)
(3.333333333,1.772731253)(3.500000000,1.661997649)(3.666666667,1.572474418)
(4.,1.437615655)(4.333333333,1.341880063)(4.500000000,1.304001198)(4.666666667,1.271161322)
(5.,1.217280511)(5.333333333,1.175201862)(5.500000000,1.157524232)(5.666666667,1.141671934)
(6.,1.114504301)(6.333333333,1.092180588)(6.500000000,1.082483786)(6.666666667,1.073616703)
(7.,1.058020031)(8.,1.023789020)(9.,1.001672056)(10.,.9867348137)(11.,.9763238359)
\pscurve[linecolor=blue](5.1,20)(5.333333333,10.59524261)(5.500000000,7.983653347)
(5.666666667,6.431156556)(6.,4.691963218)(6.333333333,3.754076768)
(6.500000000,3.432844705)(6.666666667,3.173379180)(7.,2.781104818)
(8.,2.125209579)(9.,1.801466632)(10.,1.611199548)(11.,1.487160306)
\pscurve[linecolor=purple](6.23,20)(6.333333333,15.44277038)(6.500000000,11.07829450)(6.666666667,8.656808972)(7.,6.084720541)(7.333333333,4.755143336)
(7.5,4.308434612)
(7.666666667,3.950707522)(8.,3.414902375)(8.5,2.882482960)(9.,2.531338080)
(10.,2.099814608)(11.,1.846979686)
\pscurve[linecolor=black](6.39,20)(6.500000000,14.46290919)(6.666666667,10.62707361)(7.,7.)(7.333333333,5.288415739)(7.5,4.736866769)(7.666666667,4.303492379)(8.,3.668016173)(8.5,3.052547604)(9.,2.655264456)(10.,2.176376408)
(11.,1.900552860)
\rput(11.5,0.5){$c$}
\rput(0.4,20){$\sigma$}
 \rput(1.6,20.5){$E_1$}
 \rput(5.2,20.5){$E_2$}
 \rput(6.02,20.5){$E_3$}
 \rput(6.76,20.5){$E_4$}
\rput(1.5,1.5){\tiny$\left(1,1\right)$}
\rput(3.3,6.5){\tiny$\left(15,8\right)$}
\rput(6.7,4.5){\tiny$\left(57,8\right)$}
\rput(6.8,12.5){\tiny$\left(127,8\right)$}
\rput(8.7,5.5){\tiny$\left(127,7\right)$}
 \end{pspicture*}}
\subfigure[$n=8$]{
\begin{pspicture*}(0,0)(6.7,5.5)
\psset{xunit=0.55cm, yunit=0.25cm}
\psaxes{->}(0,0)(11,20)
\pscurve[linecolor=red](1.726,20)
(1.750000000,18.80588525)(2.,7.847547601)(2.250000000,4.646205318)
(2.333333333,4.079417777)(2.500000000,3.289741116)(2.666666667,2.775644710)
(2.750000000,2.582770669)(3.,2.163221300)(3.250000000,1.891333137)
(3.333333333,1.821432173)(3.500000000,1.703665496)(3.666666667,1.60868188)
(3.750000000,1.567879296)(4.,1.465991139)(5.,1.233889812)(5.333333333,1.189699024)
(5.500000000,1.171144344)(5.666666667,1.154509890)(6.,1.126009141)(6.333333333,1.102593805)
(6.500000000,1.092422787)(6.666666667,1.083121484)(7.,1.066758241)(7.333333333,1.052882545)
(7.500000000,1.046725519)(8.,1.030815000)(8.500000000,1.017995538)(9.,1.007541595)(9.500000000,.9989296509)(10.,.9917737023)
(11.,.9807383124)
\pscurve[linecolor=blue](5.446,20)(5.500000000,17.23944525)
(5.666666667,11.45604186)(6.,6.891207384)(6.333333333,5.002752569)(6.500000000,4.428142515)(6.666666667,3.988041222)
(7.,3.360764010)(7.333333333,2.937306734)(7.500000000,2.773816625)(8.,2.405904364)(8.500000000,2.155256926)(9.,1.974344572)(9.500000000,1.838104328)
(10.,1.732114771)(11.,1.578557606)
\pscurve[linecolor=purple](6.946,20)(7.,17.63500970)(7.333333333,9.727549724)
(7.500000000,7.977407307)(8.,5.273807085)(8.500000000,4.030882059)
(9.,3.324624566)(9.500000000,2.872234232)(10.,2.559121014)(11.,2.156255469)
\pscurve[linecolor=black](7.415,20)(7.500000000,16.45237470)(8.,8.)(8.500000000,5.402702428)
(9.,4.166436205)(9.500000000,3.450604423)(10.,2.986528199)
(11.,2.424375951)
\rput(11.5,0.5){$c$}
\rput(0.4,20){$\sigma$}
 \rput(1.5,20.5){$E_1$}
 \rput(5.4,20.5){$E_2$}
 \rput(6.8,20.5){$E_3$}
 \rput(7.9,20.5){$E_4$}
\rput(1.5,1.5){\tiny$\left(1,1\right)$}
\rput(3.3,4.5){\tiny$\left(17,9\right)$}
\rput(6.9,5.5){\tiny$\left(73,9\right)$}
\rput(7.6,9.5){\tiny$\left(185,9\right)$}
\rput(9.7,5.5){\tiny$\left(255,8\right)$}
\end{pspicture*}}
\caption{$c$--$\sigma$ graphs for $n=3,4,5,6,7,8$}\label{cs}
\end{figure}
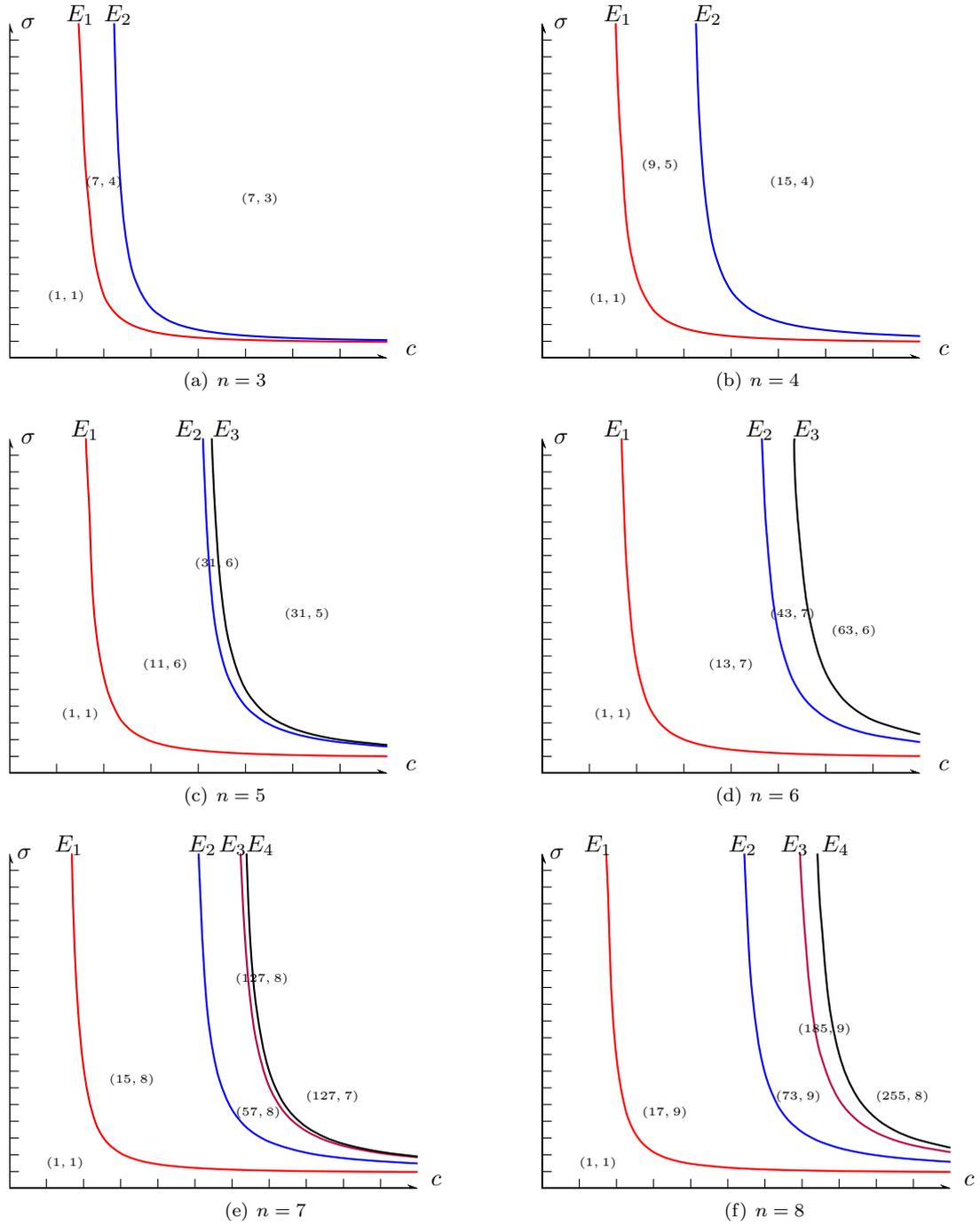

In this section, we measure how much improvement is provided by the  special
algorithm  over the general  algorithm. We use the model for simultaneous
decision in Example~\ref{models} as a benchmark.
In order to measure the performance, we first need to fix the implemental
details of several steps. We have made the following choices.
\begin{enumerate}
\item[(1)]In {\bf Algorithm \ref{CP}. Lines \ref{DECP2}} and {\bf \ref{NDECP2}
}, 
we use  the command {\tt BorderPolynomial}  in {\tt DISCOVERER} \cite{x}
to compute the projection of parametric polynomial equations, which is based
on triangular decomposition method.
\item[(2)]In {\bf Algorithm \ref{EC}. Lines \ref{nde}, \ref{nsde}, \ref{lnnd},
\ref{nscondition0}, \ref{nscondition1}} and {\bf \ref{nscondition2}}, 
    we first cancel the denominators. It is safe due to the condition (3)
in Definition \ref{cd}. Then we use {\tt  RootFinding[Isolate]}  in {\tt
Maple16} 
to compute the real solutions of polynomial equations and inequations.
\end{enumerate}

\medskip
\noindent In the following, we provide the experimental results in three
figures:
Figure~\ref{timings}, Figure~\ref{time} and Figure~\ref{cs}.
\medskip

\begin{itemize}
\itemsep=1em

\item
Figure~\ref{timings} provides the timing comparison of Algorithm~\ref{sa}
(Section~\ref{algorithm}) and the general algorithm (Section~\ref{review})
for
 $n=2,\ldots,15$ and $c=1,\ldots,15$.  The top entries are the timings in
seconds for Algorithm~\ref{sa} and the bottom entries are for the general
algorithm.
 The symbol $\infty$ means the computational time
is greater than $1500$ seconds (aborted).
 Both programs were written in Maple and were executed on  an  Intel Core
i7 processor (2.3GHz CPU, 4 Cores and 8GB total memory).

\medskip \noindent
Observe that Algorithm \ref{sa} performs much faster  than the general algorithm
for $n\geq 3$. As is pointed out by \cite{n2012}, when $n>5$, it becomes
expensive for the general algorithm to compute the Hurwitz determinants and
the sizes of these determinants are usually huge, which leads to much difficulties
of the subsequent computations. Moreover, when~$c$ is relatively large, the
real solution isolation of the general algorithm performs quite slowly, even
needs thousands of seconds for one sample point. 

\medskip \noindent
Note also that the special algorithm is a bit slower than the general algorithm
when $n=2$.
 The main reasons are that the special algorithm benefit little from  exploiting
the special structure and that the special algorithm pays the   overhead
cost for analyzing the structure. \item
Figure~\ref{time} provides the timings of Algorithm~\ref{sa}
as a graph  over $time $ and $(n,c)$.
By fitting, we find that it is very close to the graph of
$$time \approx 0.012\left(n-2\right)e^{0.6c}.$$

\noindent
Observe that the computational time is approximately linear with respect
to $n$ (the number of proteins ) and  exponential with respect to $c$ (the
cooperativity).

\item
Figure~\ref{cs}
provides, for $n=3,\ldots,8$, the partition of the $c$-$\sigma$ plane into
several cells by several curves  $E_i(c,\sigma)=0$.  In each cell,  the 
number of (stable) equilibriums is uniform (presented in each cell). Note
that Algorithm \ref{sa} can be applied to  rational $c$ values. For each
$n$, we computed all the critical
$\sigma$ values for different rational $c$ values, obtaining sufficiently
many $(c, \sigma)$ points.  Then we obtained $E_i$ by curve fitting.

\medskip \noindent
Note that we are showing a complete answer to the multistability problem
of the system for the given $n$ values.  We also remark that
 the curve $E_{\lceil \frac{n}{2}\rceil}(c,\sigma)=0$ matches $c-n+1-\left(
\frac{c}{\sigma }\right) ^{\frac{c}{c+1}}=0. $ Note that  only    when $(c,\sigma)$
is  beyond the curve, the number of stable equilibriums is $n$. Thus we have
verified the following conjecture in \cite{cd2002} for $n=3,\ldots,8$:
the system has exactly $n$ stable equilibriums if and only if $c-n+1-\left(
\frac{c}{\sigma }\right) ^{\frac{c}{c+1}}>0$.
\end{itemize}

From the computational results, one sees immediately that the equilibrium classifications
of MSRS also have certain special structures, with interesting biological implications.  
A detailed analysis of the structures and their biological implications 
will be reported in a forthcoming article.

\end{document}